\documentclass[a4paper,10pt]{article}

\usepackage{bm}
\usepackage{verbatim}
\usepackage[colorlinks=false]{hyperref} 
\usepackage{subfigure}
\usepackage{amsmath,amsthm,amsfonts,amssymb,graphics,graphicx,epsfig,color,times}
\usepackage{pgfplots}
\usepackage{epstopdf}
\usepackage{enumerate}

\def\picill#1by#2(#3)
{\vbox to #2 {\hrule width #1 height 0pt depth 0pt
\vfill\epsffile{#3}}}

\newcommand{\eq}{\begin{equation}}
\newcommand{\en}{\end{equation}}
\newcommand{\eqa}{\begin{eqnarray}}
\newcommand{\ena}{\end{eqnarray}}
\newcommand{\tr}{\mathrm{tr}}
\newtheorem{theorem}{Theorem}
\newtheorem{corollary}{Corollary}

\begin{document}

\setlength{\unitlength}{1mm}

\thispagestyle{empty}


 \begin{center}
{\large{\bf  Quantum teleportation and Birman--Murakami--Wenzl algebra}}
  \\[4mm]

{\small Kun Zhang \footnote{kun\_zhang@whu.edu.cn} and Yong Zhang \footnote{yong\_zhang@whu.edu.cn}}\\[3mm]

Center for Theoretical Physics, Wuhan University, Wuhan 430072, P. R. China

School of Physics and Technology, Wuhan University, Wuhan 430072, P. R. China \\[1cm]

\end{center}

\vspace{0.2cm}

\begin{center}
\parbox{13cm}{
\centerline{\small  \bf Abstract}  \noindent

 In this paper, we investigate the relationship of  quantum teleportation in quantum information science and
 the Birman--Murakami--Wenzl (BMW) algebra in low-dimensional topology. For simplicity, we
 focus on the two spin-1/2 representation of the BMW algebra, which is generated by both the Temperley--Lieb projector
 and the Yang--Baxter gate. We describe quantum teleportation using the Temperley--Lieb projector and the Yang--Baxter gate,
 respectively, and study teleportation-based quantum computation using the Yang--Baxter gate. On the other hand,  we exploit
 the extended Temperley--Lieb diagrammatical approach to clearly show that the tangle relations of  the BMW algebra have a natural
 interpretation of quantum teleportation. Inspired by this interpretation, we construct a general representation of the tangle relations
 of the BMW algebra and obtain  interesting representations of the BMW algebra. Therefore our research sheds a light on a
 link between quantum information science and low-dimensional topology.

}

\end{center}

\vspace{.2cm}

\begin{tabbing}

\end{tabbing}


\section{Introduction}

Quantum entanglement \cite{NC2011, Preskill97} is the key reason why the quantum information processing outperforms the classical information processing, and has been widely exploited in quantum information and computation science. With the help of quantum entanglement, an unknown quantum state can be transported from one place to another place, dubbed quantum teleportation \cite{BBCCJPW93, Vaidman94, BDM00, Werner01}. On the other respect, there are natural similarities between quantum entanglement and topological entanglement \cite{Aravind97, KL02}, where the latter characterizes topological configurations of links or knots \cite{Kauffman02}. Nontrivial unitary solutions of the Yang--Baxter
equation \cite{YBE67}, called the Yang--Baxter gates, have been introduced to clarify such similarities \cite{Dye03, KL04, ZKG04, AJJ16}. They can detect knots or links and can be viewed as quantum gates to perform universal quantum computation as well. Note that a thorough understanding about a relation between quantum entanglement and topological entanglement remains unclear.

The Yang--Baxter equation \cite{YBE67} arose in the study of both 1+1-dimensional quantum many-body systems and vertex models in statistics physics, and its solution naturally gives rise to a representation of the braid group describing links or knots \cite{Kauffman02}. Both the  Temperley--Lieb algebra \cite{TL71} and the  BMW algebra \cite{BW89} are exploited
in the systematic construction of solutions of the Yang--Baxter equation  \cite{Jones89,CGX91,WXSZHW10}. The Temperley--Lieb algebra is associated with solutions of the Yang--Baxter equation with two distinctive eigenvalues and is related to the Jones polynomial in knot theory \cite{Kauffman02}, whereas the BMW algebra is associated with solutions of the Yang--Baxter equation with three distinctive eigenvalues and is related to the Kauffman polynomial in knot theory \cite{Kauffman02}.

Previous research has extensively studied the relationship between quantum teleportation, the Temperley--Lieb algebra and the Yang--Baxter equation \cite{Zhang06,ZZP15}. The extended Temperley--Lieb diagrammatical approach \cite{Zhang06,ZZP15} is devised to characterize topological features of quantum entanglement and quantum teleportation. However, the former research either concentrates on the topic of quantum teleportation using the Yang--Baxter gate or on the topic of quantum teleportation using the Temperley--Lieb projector. Since a representation of the BMW algebra is generated by both the Yang--Baxter gate and the Temperley--Lieb projector, there remains a natural question to be answered: what about quantum teleportation using the BMW algebra? In this paper, we investigate this problem  and expect to find something novel which is not presented in \cite{Zhang06,ZZP15}.

In the two spin-1/2 representation of the BMW algebra \cite{WXSLZ15,BMW11}, we view the Temperley--Lieb projector as a two-qubit quantum measurement operator and the unitary braid representation as a two-qubit entangling gate (the Yang--Baxter gate). We show that the Temperley--Lieb projector and the Yang--Baxter gate are capable of performing the quantum teleportation protocol, respectively. Besides, we study the teleportation-based quantum computation \cite{GC99,Nielsen03} using the Yang--Baxter gate. Furthermore, we realize that the tangle relations defining the BMW algebra involving both the Temperley--Lieb projector and the Yang--Baxter gate give rise to the teleportation protocol directly. Moreover, we are able to construct interesting representations of the BMW algebra in the extended Temperley--Lieb diagrammatical approach \cite{Zhang06,ZZP15} different from the representation
in \cite{WXSLZ15,BMW11}.

 The paper is organized as follows. Section \ref{section review BMW algebra} reviews the two spin-1/2 realization of the BMW algebra.  Section \ref{section teleportation TL} and \ref{section teleportation YBG} perform the quantum teleportation protocol via the Temperley--Lieb projector and the Yang--Baxter gate, respectively. Section \ref{section teleportation quantum computation YBG} is about the teleportation-based quantum computation set by the Yang--Baxter gate. Section \ref{section interpretation BMW teleportation} presents the quantum teleportation interpretation of the tangle relations in the BMW algebra. Section \ref{section construction rep from tangle relation} describes a  method of constructing a general representation of the tangle relations of the BMW algebra. Section \ref{section concluding remarks} is on concluding remarks. Four appendices are added to complete the paper.  Appendix \ref{appendix Brauer algebra} reviews a relation between the Brauer algebra and quantum teleportation \cite{Zhang06}, since the BMW algebra is a deformation of the Brauer algebra \cite{Brauer37}. Appendix \ref{appendix further study TL B} collects the extended Temperley--Lieb configurations of the Yang--Baxter gate which are not in the main context of the paper. Appendix \ref{appendix trial solution tangle relation} shows the procedure of constructing interesting representations of the BMW algebra.
 Appendix \ref{skew transpose} introduces the new conventions and notations to simplify some complicated algebraic relations in the paper.

\section{Review on the BMW algebra}

\label{section review BMW algebra}

In this section, we make a brief sketch on the definition of the BMW algebra \cite{BW89} and its two spin-1/2 representation \cite{WXSLZ15} in the viewpoint of quantum
information and computation. Meanwhile, we set up notations and conventions for the whole paper. In addition, we make a simple study on the Brauer algebra \cite{Brauer37}
in Appendix \ref{appendix Brauer algebra}  whose deformation is the BMW algebra.

The BMW algebra $B_n(w,\sigma)$ \cite{BW89} contains two complex parameters $w$ and $\sigma$, and it has two types of generators: the one denoted as $e_i$ associated with the Temperley--Lieb algebra \cite{TL71} and the other one denoted as $b_i$ associated with the braid group \cite{Kauffman02}, with $i=1, \ldots, n-1$.
The Temperley--Lieb idempotents $e_i$ satisfy the defining relations of the Temperley--Lieb algebra,
\eq
\label{definition BMW algebra_TL}
  \begin{array}{ll}
    e_{i}^2=e_{i}, &  e_{i}e_{i\pm 1}e_{i}=d^{-2}e_{i}, \\
    e_ie_j=e_j e_i, & |i-j|\geq 2,
  \end{array}
\en
with $d$ called the loop parameter. The braid generators $b_i$ satisfy the defining relations of the braid group,
\eq
\label{definition BMW algebra_YBE}
  \begin{array}{ll}
    b_{i}b_{i\pm 1}b_{i}=b_{i\pm 1}b_{i}b_{i\pm 1}, &  \\
    b_ib_j=b_jb_i, & |i-j|\geq 2.
  \end{array}
\en
The first type of the mixed relations between the Temperley--Lieb idempotents $e_i$ and the  braid generators $b_i$ are given by
\eq
\label{definition BMW algebra_Mixed}
  \begin{array}{ll}
    b_{i}-b_{i}^{-1}=w (1\!\!1-de_i), &  \\
    e_{i}b_i=b_{i}e_{i}=\sigma e_i, &  \\
    b_{i\pm 1}e_{i}b_{i\pm 1}=b_{i}^{-1}e_{i\pm 1}b_{i}^{-1}, &
  \end{array}
\en
where $1\!\!1$ denotes the identity operator, and the second type of the mixed relations are given by
\eq
\label{definition BMW algebra_tangle}
  \begin{array}{ll}
    b_{i\pm 1}b_{i}e_{i\pm 1}=e_{i}b_{i\pm 1}b_{i}=de_{i}e_{i\pm 1},
  \end{array}
\en
which are called the tangle relations in this paper. Note that there is a constraint relation among the three parameters $d, \sigma, w$
given by $d=1-(\sigma-\sigma^{-1})/w$.

A tensor product representation of the BMW algebra $B_n(w,\sigma)$ can be in general constructed  in the following way. The Temperley--Lieb idempotents
$e_i$ and the braid generators $b_i$ assume the respective forms,
\eqa
e_i &=& 1\!\!1^{\otimes(i-1)}\otimes E\otimes 1\!\!1^{\otimes(n-i-1)},\\
b_i &=& 1\!\!1^{\otimes(i-1)}\otimes B\otimes 1\!\!1^{\otimes(n-i-1)},
\ena
where the symbol $E$ is called the Temperley--Lieb matrix and the symbol $B$ is called the braid matrix. It means that the Temperley--Lieb idempotents
$e_i$ and the braid generators $b_i$ are acting on the vector space of the $i$th and $i+1$th sites. Furthermore, the braid matrix $B$ has three distinctive eigenvalues denoted as $\lambda_1$, $\lambda_2$ and $\lambda_3$ with the constraint relation $\lambda_2\lambda_3=-1$. The parameter $\sigma$ is set as $\sigma=\lambda_1$, and the parameter $w$ is set as
$w=\lambda_2+\lambda_3$.


 A spin-1/2 representation is characterized by the two-dimensional Hilbert space $\mathcal H_2$ with the basis states $|0\rangle$ and $|1\rangle$.  In convention, the state $|0\rangle$ stands for spin up and the state $|1\rangle$ for spin down. A single qubit in quantum information and computation \cite{NC2011, Preskill97} can be physically realized
 as a spin-1/2 representation, and its quantum state has the form $|\alpha\rangle=a|0\rangle+b|1\rangle$ with complex numbers $a$ and $b$ satisfying $|a|^2+|b|^2=1$. A quantum gate is defined as a unitary transformation acting on qubits, for example, a single qubit gate including the identity operator $1\!\!1_2$, the Pauli $X$ gate and the Pauli $Z$ gate
 given by
\eq
\label{definition Pauli gates}
1\!\! 1_2 =\left(\begin{array}{cc}
              1 & 0 \\
              0 & 1 \\
            \end{array}\right), \quad X=\left(\begin{array}{cc}
              0 & 1 \\
              1 & 0 \\
            \end{array} \right), \quad  Z=\left(\begin{array}{cc}
              1 & 0 \\
              0 & -1 \\
            \end{array} \right).
\en
Note that $Z|i\rangle=(-1)^i|i\rangle$ with $i=0, 1$ and the Pauli $Y$ gate is defined as $Y=ZX$.

A two spin-1/2 representation, which can be recognized as a physical realization of a two-qubit system, is described by the four-dimensional Hilbert space $\mathcal H_2\otimes \mathcal H_2$ with the tensor product basis states $|i,j\rangle$ or $|ij\rangle$ with $i,j=0,1$. An example for the two spin-1/2 representation of the BMW algebra has been shown
in \cite{WXSLZ15}, in which the associated Temperley--Lieb matrix $E$ has the form
\eq
\label{TL matrix of BMW algebra}
E = \frac 1 4\left(
             \begin{array}{cccc}
               1 & ie^{-i\phi} & ie^{-i\phi} & e^{-2i\phi} \\
               -ie^{i\phi} & 1 & 1 & -ie^{-i\phi} \\
               -ie^{i\phi} & 1 & 1 & -ie^{-i\phi} \\
               e^{2i\phi} & ie^{i\phi} & ie^{i\phi} & 1 \\
             \end{array}
           \right),
\en
with the real number $\phi$, and the associated braid matrix $B$ has the form
\eq
\label{braid matrix of BWM algebra}
B = \frac {e^{i\frac {3\pi}{4}}}{2}\left(
                                                            \begin{array}{cccc}
                                                              1 & -e^{-i\phi} & -e^{-i\phi} & -e^{-2i\phi} \\
                                                              e^{i\phi} & 1 & -1 & e^{-i\phi} \\
                                                              e^{i\phi} & -1 & 1 & e^{-i\phi} \\
                                                              -e^{2i\phi} & -e^{i\phi} & -e^{i\phi} & 1 \\
                                                            \end{array}
                                                          \right),
\en
 with three distinctive eigenvalues $\lambda_1=e^{i5\pi/4}$, $\lambda_2=e^{i3\pi/4}$ and $\lambda_3=e^{i\pi/4}$, $\lambda_2$ being double-degenerated.
  The parameters $\sigma$, $w$ and $d$ in the BMW algebra $B_n(w,\sigma)$ can be calculated as
\eq
\sigma=\lambda_1=e^{i5\pi/4}, \quad w=\lambda_2+\lambda_3=\sqrt 2 i, \quad d=1-(\sigma-\sigma^{-1})/w=2.
\en

 In this paper, since the Temperley--Lieb matrix $E$ satisfies the idempotent relation $E^2=E$, it is viewed as  a two-qubit projective measurement operator.
 On the other respect, the braid matrix $B$ is unitary satisfying $B^\dag B=BB^\dag=1\!\!1_4$, and thus it is a two-qubit quantum gate \cite{NC2011, Preskill97},
 called the Yang--Baxter gate \cite{Zhang06} in the following study.

\section{Quantum teleportation using the Temperley--Lieb projector}

\label{section teleportation TL}

In this section, we construct a complete set of two-qubit projective measurement operators with the Temperley--Lieb projector $E$ (\ref{TL matrix of BMW algebra})
as a special case, and then exploit such a set to describe the quantum teleportation protocol in the extended Temperley--Lieb diagrammatical approach \cite{Zhang06, ZZP15}.

\subsection{The Temperley--Lieb projector as a two-qubit measurement operator}

\label{subsection Temperley--Lieb projector}

In quantum information and computation \cite{NC2011, Preskill97}, the well-known complete set of two-qubit projective measurement operators is formed by
the Bell states \cite{ZZP15, ZZ14} given by
\eq
\label{definition four Bell states}
|\psi(ij)\rangle=(1\!\!1_2\otimes W_{ij})|\Psi\rangle, \quad i,j=0,1,
\en
with the  single-qubit gate $W_{ij}=X^iZ^j$ and the EPR pair $|\Psi\rangle=\frac 1 {\sqrt 2}(|00\rangle+|11\rangle)$, and the Bell projective measurement
operator is denoted by $|\psi(ij)\rangle\langle\psi(ij)|$. Note that the Bell states $|\psi(ij)\rangle$ form an
orthonormal basis for the two-qubit Hilbert space $\mathcal H_2\otimes \mathcal H_2$, named the Bell basis. The orthonormal condition for the Bell basis (or
the orthonormal condition for the Bell projective measurement operators) is given by
\eq
\langle\psi(i_2j_2)|\psi(i_1j_1)\rangle=\delta_{i_1i_2}\delta_{j_1j_2},
\en
with the Kronecker delta function $\delta_{ij}=1$ for $i=j$ and $\delta_{ij}=0$ for $i\neq j$, which is equivalent to
\eq
\frac 1 2\tr\left(W^\dag_{i_2j_2}W_{i_1j_1}\right)=\delta_{i_1i_2}\delta_{j_1j_2},
\en
with the index $\dag$ denoting the Hermitian conjugation.

In view of the construction of the Bell projective measurement operators $|\psi(ij)\rangle\langle\psi(ij)|$, we reformulate the  Temperley--Lieb
projector $E$ (\ref{TL matrix of BMW algebra})
as a two-qubit projective operator,
\eq
 \label{definition E 00 projector}
 E=|\Psi_{M_{00}}\rangle\langle\Psi_{M_{00}}|,
 \en
where the Bell-like state $|\Psi_{M_{00}}\rangle$ denotes  the EPR pair $|\Psi\rangle$ with the local action of the single-qubit gate $M_{00}$,
 \eq
 \label{definition M 00 state}
 |\Psi_{M_{00}}\rangle=(1\!\!1_2\otimes M_{00})|\Psi\rangle.
 \en
After calculation, the single-qubit gate $M_{00}$ can be expressed as a product of elementary single-qubit gates, $M_{00}=R_\phi HSH R_{\phi}$.
The symbols $H$, $S$ and $R_\phi$ stand for the Hadamard gate, the phase gate and the phase shift gate \cite{NC2011, Preskill97}, respectively,
and they are given by
\eq
\label{definition H S R gates}
H=\frac 1 {\sqrt 2}\left(
                     \begin{array}{cc}
                       1 & 1 \\
                       1 & -1 \\
                     \end{array}
                   \right), \quad
S=\left(
                     \begin{array}{cc}
                       1 & 0 \\
                       0 & i \\
                     \end{array}
                   \right), \quad
R_\phi= \left(
                     \begin{array}{cc}
                       1 & 0 \\
                       0 & e^{i\phi} \\
                     \end{array}
                   \right).
\en
Note that both the Pauli $Z$ gate (\ref{definition Pauli gates}) and the phase gate $S$ are special cases of the phase shift gate $R_\phi$, namely $Z=R_{\pi}$ and $S=R_{\pi/2}$.

Furthermore, a complete set of two-qubit projective measurement operators $E_{ij}$ including the  Temperley--Lieb
projector $E$ (\ref{TL matrix of BMW algebra}) as a special case, $E=E_{00}$, can be constructed in the way
\eq
\label{definition E ij projector}
E_{ij}=|\Psi_{M_{ij}}\rangle\langle\Psi_{M_{ij}}|, \quad i,j=0,1,
\en
with $|\Psi_{M_{ij}}\rangle=(1\!\!1_2\otimes M_{ij})|\Psi\rangle$,  where the single-qubit gates $M_{ij}$ have the form
\eq
\label{definition M ij gates}
M_{00}=R_\phi HSH R_{\phi},\quad M_{01}=R_\phi ZR_\phi,\quad M_{10}=XZ,\quad M_{11}=R_\phi HS^\dag H R_\phi.
\en
Note that the Bell-like states $|\Psi_{M_{ij}}\rangle$ form an orthonormal basis for the two-qubit Hilbert space $\mathcal H_2\otimes \mathcal H_2$ with
the orthonormal condition given by
\eq
E_{i_1j_1}E_{i_2j_2}=\delta_{i_1i_2}\delta_{j_1j_2}E_{i_1j_1},
\en
equivalent to
\eq
\label{relation orthonormal condition M}
\frac 1 2\tr\left(M^\dag_{i_2j_2}M_{i_1j_1}\right)=\delta_{i_1i_2}\delta_{j_1j_2}.
\en

Here we make remarks about both the Bell states $|\psi(ij)\rangle$ and the Bell-like states $|\Psi_{M_{ij}}\rangle$. First, both states can be exactly determined by
the complete basis of unitary operators $W_{ij}$ (or $M_{ij}$) \cite{VW00}. Second, they are  maximally entangled states in quantum information science
\cite{NC2011, Preskill97}. Third, all of the associated two-qubit projective measurement operators are able to generate the representation of the Temperley--Lieb algebra \cite{Zhang06}, and thus all of them can be called the Temperley--Lieb algebra projector.

\subsection{Extended Temperley--Lieb configuration of quantum teleportation}

\label{subsection TL teleportation}

In the extended Temperley--Lieb diagrammatical approach \cite{Zhang06, ZZP15}, the Bell state $|\Psi\rangle$ is pictured as a cup configuration and its complex
conjugation $\langle\Psi|$ is pictured as a cap configuration; a single-qubit gate acting on the Bell states is pictured as a solid point on associated configurations.
Thus the Bell projective measurement operators $|\psi(ij)\rangle\langle\psi(ij)|$ have a diagrammatical representation shown as
\eqa
\label{TL Bell measurement}
\setlength{\unitlength}{0.6mm}
\begin{array}{c}
\begin{picture}(65,30)
\put(62,18){\line(0,1){12}}
\put(52,18){\line(0,1){12}}
\put(52,18){\line(1,0){10}}
\put(62,24){\circle*{2.}}
\put(64.,23){\tiny{$W_{ij}$}}
\put(6,14){$|\psi(ij)\rangle \langle\psi(ij)|=$}
\put(52,0){\line(0,1){12}}
\put(62,0){\line(0,1){12}}
\put(52,12){\line(1,0){10}}
\put(62,6){\circle*{2.}}
\put(64,5){\tiny{$W_{ij}^\dag$}}
\end{picture}
\end{array}
\ena
where the diagram is read from the bottom to the top corresponding to the convention that the algebraic expression is read from the right to the left.
Furthermore, the solid point representing a single-qubit gate $U$ on the cup or cap configuration is able to flow from the one branch to the other branch
with an additional transpose on such the gate $U$, shown in
\eqa
\setlength{\unitlength}{0.6mm}
\begin{array}{c}
\begin{picture}(45,16)
\put(2,0){\line(0,1){12}}
\put(12,0){\line(0,1){12}}
\put(2,0){\line(1,0){10}}
\put(12,6){\circle*{2.}}
\put(14.,5){\tiny{$U$}}
\put(23.,6){$=$}
\put(32,5){\tiny{$U^T$}}
 \put(40.8,6){\circle*{2.}}
\put(40.8,0){\line(0,1){12}}
\put(50.8,0){\line(0,1){12}}
\put(40.8,0){\line(1,0){10}}
\end{picture}
\end{array}
\ena
which is related to the algebraic formula
\eq
\label{relation flowing gate Bell state}
(1\!\!1_2\otimes U)|\Psi\rangle=(U^T\otimes 1\!\!1_2)|\Psi\rangle,
\en
with the symbol $T$ denoting the matrix transpose. Note that the property that a single-qubit gate flows on the configuration plays a crucial role in the
extended Temperley--Lieb diagrammatical approach to quantum teleportation \cite{Zhang06, ZZP15}.

In quantum teleportation \cite{BBCCJPW93, Vaidman94, BDM00, Werner01}, Alice has an unknown qubit $|\alpha\rangle$ to be transmitted to Bob and meanwhile shares
the Bell state $|\Psi\rangle$ with Bob, namely, Alice and Bob prepare the state $|\alpha\rangle\otimes|\Psi\rangle$. Such the initial state has the
extended Temperley--Lieb configuration,
\eqa
\setlength{\unitlength}{0.6mm}
\begin{array}{c}
\begin{picture}(70,16)
\put(10,7){$|\alpha\rangle\otimes |\Psi\rangle=$}
\put(45,0){\makebox(4,4){$\nabla$}}
\put(47,3.7){\line(0,1){10.3}}
\put(57,0){\line(0,1){14}}
\put(67,0){\line(0,1){14}}
\put(57,0){\line(1,0){10}}
\end{picture}
\end{array}
\ena
where the qubit $|\alpha\rangle$ is depicted as a vertical line with the symbol $\nabla$ at the bottom. Now Alice
performs the Bell projective measurement $|\psi(ij)\rangle\langle\psi(ij)|\otimes 1\!\!1_2$ on her qubits, which
is illustrated in the extended Temperley--Lieb configuration,
\eqa
\label{TL teleportation}
\setlength{\unitlength}{0.6mm}
\begin{array}{c}
\begin{picture}(80,45)
\put(10,0){\makebox(4,4){$\nabla$}}
\put(12.,4){\line(0,1){20}}
\put(22,18){\circle*{2.}}
\put(23,17){\tiny{$W_{ij}^\dag$}}
\put(22,0){\line(0,1){24}}
\put(32,0){\line(0,1){42}}
\multiput(8,12)(1,0){30}{\line(1,0){.5}}
\put(22,0){\line(1,0){10}}
\put(12,24){\line(1,0){10}}
\put(12,30){\line(1,0){10}}
\put(12,30){\line(0,1){12}}
\put(22,30){\line(0,1){12}}
\put(22,36){\circle*{2.}}
\put(23,35){\tiny{$W_{ij}$}}
\put(40,21){\makebox(14,10){$=\frac 1 2$}}
\put(62,30){\line(1,0){10}}
\put(62,30){\line(0,1){12}}
\put(72,30){\line(0,1){12}}
\put(72,36){\circle*{2.}}
\put(73,35){\tiny{$W_{ij}$}}
\put(80,0){\makebox(4,4){$\nabla$}}
\put(82,4){\line(0,1){38}}
\put(82,21){\circle*{2.}}
\put(84,20){\tiny{$W_{ij}$}}
\end{picture}
\end{array}
\ena
with the dashed line denoting the time boundary between the initial state and the Bell measurements. On the left of the diagram, the identity operator $1\!\!1_2$ is drawn as the single vertical line; after the single-qubit gate $W^\dag_{ij}$ flows from Alice's system to Bob's system with the transpose, $(W^\dag_{ij})^T=W_{ij}$, the unknown qubit $|\alpha\rangle$   has been transferred from Alice to Bob because of the topological deformation. On the right of the diagram, the factor $1/2$ is a normalization factor contributed by a pair of vanishing cup and cap configurations. Therefore, this diagram (\ref{TL teleportation}) is related to the algebraic formalism
 \eq
 \label{teleportation equation projector}
 (|\psi(ij)\rangle\langle\psi(ij)|\otimes 1\!\! 1_2)(|\alpha\rangle\otimes |\Psi\rangle)=\frac 1 2 |\psi(ij)\rangle \otimes W_{ij}  |\alpha\rangle,
 \en
 which is formulated, with the completeness relation of Bell projective measurement operators, $\sum_{i,j=0}^1|\psi(ij)\rangle\langle\psi(ij)|=1\!\!1_4$, as
 \eq
 \label{teleportation equation original}
 |\alpha\rangle \otimes |\Psi\rangle =\frac 1 2 \sum_{i,j=0}^1  |\psi(ij)\rangle \otimes W_{ij} |\alpha\rangle,
 \en
called the teleportation equation in \cite{ZZP15, ZZ14}. Finally, Bob has to acquire the Bell measurement results labeled as $(i,j)$ from Alice in order to apply the unitary correction operations  $W^\dag_{ij}$ on his qubit to obtain the unknown state $|\alpha\rangle$. Note that the teleportation protocol of Bob sending an unknown qubit $|\alpha\rangle$
to Alice can be characterized in
  \eq
   \label{teleportation equation original transpose}
   |\Psi\rangle\otimes|\alpha\rangle  =\frac 1 2 \sum_{i,j=0}^1  W^T_{ij} |\alpha\rangle\otimes|\psi(ij)\rangle,
  \en
 called the transpose teleportation equation in \cite{ZZ14}.

  \subsection{Quantum teleportation using the Temperley--Lieb projector $E_{ij}$ (\ref{definition E ij projector})}

  In quantum teleportation \cite{BBCCJPW93, Vaidman94, BDM00, Werner01}, we are allowed to replace the initial maximal entanglement resource $|\Psi\rangle$ with the
  Bell-like state $|\Psi_{M_{00}}\rangle$ (\ref{definition M 00 state}) and replace the Bell projective measurement operator $|\psi(ij)\rangle\langle\psi(ij)|$ with
  the Bell-like projective operator $E_{ij}$ (\ref{definition E ij projector}). Similar to the teleportation configuration (\ref{TL teleportation}) using $|\Psi\rangle$
  and $|\psi(ij)\rangle\langle\psi(ij)|$, the teleportation of an unknown qubit $|\alpha\rangle$ from Alice to Bob using $|\Psi_{M_{00}}\rangle$ and $E_{ij}$
  has the extended Temperley--Lieb configuration given by
\eqa
\label{TL teleportation M 00}
\setlength{\unitlength}{0.6mm}
\begin{array}{c}
\begin{picture}(80,45)
\put(32,6){\circle*{2.}}
\put(33,5){\tiny{$M_{00}$}}
\put(10,0){\makebox(4,4){$\nabla$}}
\put(12.,4){\line(0,1){20}}
\put(22,18){\circle*{2.}}
\put(23,17){\tiny{$M_{ij}^\dag$}}
\put(22,0){\line(0,1){24}}
\put(32,0){\line(0,1){42}}
\multiput(8,12)(1,0){30}{\line(1,0){.5}}
\put(22,0){\line(1,0){10}}
\put(12,24){\line(1,0){10}}
\put(12,30){\line(1,0){10}}
\put(12,30){\line(0,1){12}}
\put(22,30){\line(0,1){12}}
\put(22,36){\circle*{2.}}
\put(23,35){\tiny{$M_{ij}$}}
\put(40,21){\makebox(14,10){$=\frac 1 2$}}
\put(62,30){\line(1,0){10}}
\put(62,30){\line(0,1){12}}
\put(72,30){\line(0,1){12}}
\put(72,36){\circle*{2.}}
\put(73,35){\tiny{$M_{ij}$}}
\put(80,0){\makebox(4,4){$\nabla$}}
\put(82,4){\line(0,1){38}}
\put(82,21){\circle*{2.}}
\put(84,20){\tiny{$M_{00}M^*_{ij}$}}
\end{picture}
\end{array}
\ena
corresponding to the algebraic formula
\eq
\label{teleportation equation half}
(|\Psi_{M_{ij}}\rangle\langle\Psi_{M_{ij}}|\otimes 1\!\! 1_2)(|\alpha\rangle\otimes |\Psi_{M_{00}}\rangle)
=\frac 1 2 |\Psi_{M_{ij}}\rangle \otimes M_{00}M^*_{ij}  |\alpha\rangle
\en
with the symbol $*$ denoting the complex conjugation,  which can be reformulated as the form of the teleportation equation,
\eq
\label{teleportation equation M 00}
 |\alpha\rangle \otimes |\Psi_{M_{00}}\rangle =\frac 1 2 \sum_{i,j=0}^1  |\Psi_{M_{ij}}\rangle \otimes M_{00}M^*_{ij} |\alpha\rangle.
 \en
 It is worth mentioning that the single-qubit gate $M_{00}$ initially acting on the Bell state $|\Psi\rangle$  has been transferred
 to Bob from Alice and has become  the single-qubit gate acting on Bob's qubit.

 Furthermore, the extended Temperley--Lieb configuration of teleportation of an unknown qubit $|\alpha\rangle$ from Bob to Alice
 can be drawn as
 \eqa
\label{TL teleportation transposed M 00}
\setlength{\unitlength}{0.6mm}
\begin{array}{c}
\begin{picture}(80,45)
\put(30,0){\makebox(4,4){$\nabla$}}
\put(32.,4){\line(0,1){20}}
\put(22,0){\line(0,1){24}}
\put(12,0){\line(0,1){42}}
\put(12,0){\line(1,0){10}}
\put(22,24){\line(1,0){10}}
\put(22,30){\line(1,0){10}}
\put(22,30){\line(0,1){12}}
\put(32,30){\line(0,1){12}}
\put(32,36){\circle*{2.}}
\put(33,35){\tiny{$M_{ij}$}}
\put(32,18){\circle*{2.}}
\put(33,17){\tiny{$M_{ij}^\dag$}}
\put(22,6){\circle*{2.}}
\put(23,5){\tiny{$M_{00}$}}
\multiput(8,12)(1,0){30}{\line(1,0){.5}}
\put(40,21){\makebox(14,10){$=\frac 1 2$}}
\put(72,30){\line(1,0){10}}
\put(72,30){\line(0,1){12}}
\put(82,30){\line(0,1){12}}
\put(82,36){\circle*{2.}}
\put(83,35){\tiny{$M_{ij}$}}
\put(60,0){\makebox(4,4){$\nabla$}}
\put(62,4){\line(0,1){38}}
\put(62,21){\circle*{2.}}
\put(64,20){\tiny{$M^T_{00}M^\dag_{ij}$}}
\end{picture}
\end{array}
\ena
which is associated with the transpose teleportation equation
\eq
 \label{teleportation equation M 00 transpose}
 |\Psi_{M_{00}}\rangle \otimes |\alpha\rangle =\frac 1 2 \sum_{i,j=0}^1 M^T_{00}M^\dag_{ij} |\alpha\rangle\otimes |\Psi_{M_{ij}}\rangle.
 \en
Comparing the equation (\ref{teleportation equation original}) and the equation (\ref{teleportation equation original transpose}), we see that the
matrix transpose is performed from $W_{ij}$ to $W^T_{ij}$. By contrast, looking at the equation (\ref{teleportation equation M 00}) and the equation
(\ref{teleportation equation M 00 transpose}), we have no transpose because of $(M_{00}M^*_{ij})^T\neq M^T_{00}M^\dag_{ij}$.

Moreover, when $i=j=0$, the single-qubit gate $M_{00}M^*_{ij}$ in the teleportation equation (\ref{teleportation equation M 00}) is identity,
$M_{00}M^*_{00}=1\!\!1_2$, and the Bell-like state $|\Psi_{M_{00}}\rangle$ can be prepared by applying the Bell-like projective measurement
operator $E_{00}$ (\ref{definition E ij projector}). Thus the teleportation of an unknown qubit $|\alpha\rangle$ from Alice to Bob can be viewed
using the Bell-like projective measurement operator $E_{00}$,
\eq
\label{teleportation equation E 00}
(E_{00}\otimes 1\!\!1_2)(|\alpha\rangle\otimes E_{00})=\frac 1 2 (|\Psi_{M_{00}}\rangle\otimes|\alpha\rangle)(1\!\!1_2\otimes\langle\Psi_{M_{00}}|),
\en
which can be derived from the equation (\ref{teleportation equation half}). Similarly, the teleportation of an unknown qubit $|\alpha\rangle$ from  Bob
to Alice can be described in the way
\eq
\label{teleportation equation E 00 transpose}
(1\!\!1_2\otimes E_{00})(E_{00}\otimes|\alpha\rangle)=\frac 1 2 (|\alpha\rangle\otimes|\Psi_{M_{00}}\rangle)(\langle\Psi_{M_{00}}|\otimes 1\!\!1_2),
\en
 where $M^T_{00}M^\dag_{00}=1\!\!1_2$ has been exploited.  Hence the operators $(E_{00}\otimes 1\!\!1_2)(1\!\!1_2\otimes E_{00})$ and
 $(1\!\!1_2\otimes E_{00})(E_{00}\otimes 1\!\!1_2)$ are capable of describing the teleportation protocol: the projector $E_{00}$ on the right works as the
 state preparation channel and the projector $E_{00}$ on the left as  the Bell measurement. In general, quantum teleportation can be performed using the
 operators $(E_{ij}\otimes 1\!\!1_2)(1\!\!1_2\otimes E_{lm})$ and $(1\!\!1_2\otimes E_{lm})(E_{ij}\otimes 1\!\!1_2)$ where $E_{ij}$ is not required the same
 as $E_{lm}$.

\section{Quantum teleportation using the Yang--Baxter gate}

\label{section teleportation YBG}

In this section, we study the application of the Yang--Baxter gate $B$ (\ref{braid matrix of BWM algebra}) to quantum teleportation. First of all, we show that
this gate can be regarded as a generalization of the Bell transform \cite{ZZ14} which is a unitary basis transformation from the product states to the Bell states.
In view of previous research \cite{ZZ14} of quantum teleportation using the Bell transform, we introduce the extended Temperley--Lieb configurations
of the Yang--Baxter gate $B$ and focus on the extended Temperley--Lieb configuration of the teleportation operator $(B\otimes 1\!\!1_2)(1\!\!1_2\otimes B)$.

\subsection{The Yang--Baxter gate $B$ (\ref{braid matrix of BWM algebra}) is the Bell transform}

\label{subsection Yang Baxter gate}

The Yang--Baxter gate $B$ (\ref{braid matrix of BWM algebra}) acting on the product states gives rise to the Bell states with the local action of
single-qubit gates  modulo a global phase factor,
\eq
\label{B gate as Bell like transform}
B|ij\rangle=e^{i\alpha_B}(R_\phi\otimes R_\phi H)|\psi(ji)\rangle,
\en
where $e^{i\alpha_B}$ is a phase factor depending on the indices $i$ and $j$ with $i,j=0,1$. Interestingly, the inverse of the Yang--Baxter gate $B$, denoted by
$B^\dag$, acting on the product basis,  also generates the Bell basis with the local action of single-qubit gates,
\eq
\label{B inverse gate as Bell like transform}
B^\dag|ij\rangle=e^{i\alpha_{B^\dag}}(R_\phi\otimes R_\phi H)|\psi(j+1,i+1)\rangle,
\en
where the factor $\alpha_{B^\dag}$ is distinctive with $\alpha_B$. Therefore, both the $B$ and $B^\dag$ gates are a generalization of the Bell transform \cite{ZZ14}
with the additional local action of single-qubit gates. For simplicity, we call the Yang--Baxter gate $B$ as the Bell transform in this paper\footnote{
The Bell transform $B_{\textit{ell}}$  in this paper is defined as
 \eq
 \label{def_bell}
 B_{\textit{ell}}=\sum_{k^\prime, l^\prime=0}^1\,e^{i\phi_{\textit{kl}}}(S_{kl}\otimes Q_{kl})|\psi(k,l)\rangle\langle
  k^\prime, l^\prime|,
\en
where $k=k(k^\prime,l^\prime)$ and $l=l(k^\prime,l^\prime)$ are bijective functions of $k^\prime$ and $l^\prime$, respectively;
$e^{i \phi_{\textit{kl}}}$ is the phase factor; and $S_{kl}$ and $Q_{kl}$ are single-qubit gates. Such the definition
of the Bell transform differs from the proposed definition of the Bell transform in previous research \cite{ZZ14} where
single-qubit gates $S_{kl}$ and $Q_{kl}$ are not involved.
}.

Obviously, the Yang--Baxter gate $B$ (or $B^\dag$) is a maximally entangling two-qubit gate \cite{NC2011, Preskill97}, since the Bell states are
maximally entangling two-qubit states and the local action of single-qubit gates does not change the entanglement property of the Bell states.
Any two-qubit gate $U$ \cite{KC02} is locally equivalent to the two-qubit gate $e^{i(aX\otimes X+bY\otimes Y+cZ\otimes Z)}$ modulo local action of single-qubit gates with
three non-local real parameters $(a, b, c)$, and the entangling power \cite{ZZF00} of the two-qubit gate $U$ can be defined as
\eq
e_p(U)=1-\cos^2 2a\cos^2 2b\cos^2 2c-\sin^2 2a\sin^2 2b\sin^2 2c
\en
ranged from 0 to 1, where $e_p(U)=1$ means that the $U$ gate is a maximally entangled two-qubit gate. After some algebra,
the non-local parameters $(a, b, c)$ for the Yang--Baxter gate $B$ take the value of $(\frac \pi 4, \frac \pi 4, 0)$, so $e_p(B)=1$.

In addition, the Yang--Baxter gate $B$ can be decomposed as a tensor product of elementary quantum gates expressed as
\eq
\label{B decomposition}
B=e^{i\frac 3 4\pi}(R_\phi\otimes R_\phi) \textit{CZ} (HZ\otimes HZ) \textit{CZ} (R_\phi^\dag\otimes R_\phi^\dag),
\en
where the \textit{CZ} gate \cite{NC2011, Preskill97} has the conventional form
\eq
\label{definition CZ gate}
\textit{CZ}=|0\rangle\langle0|\otimes 1\!\!1_2+|1\rangle\langle1|\otimes Z.
\en
The quantum circuit corresponding to such a decomposition is illustrated in
\eq
\label{diagram quantum circuit decomposition B}
  \includegraphics[width=6.8cm]{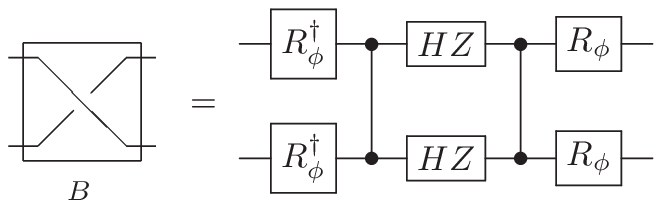}
\en
where the over-crossing feature in the box means that the $B$ gate is a braiding operator \cite{Kauffman02},
and the overall phase $e^{i\frac 3 4\pi}$ is neglected, and the configuration of two solid points connected with a vertical line
represents the \textit{CZ} gate (\ref{definition CZ gate}). Note that the decomposition (\ref{B decomposition}) of the $B$ gate using
at least two \textit{CZ} gates is optimal according to the criterion \cite{SBM04} of calculating the least number of the
\textit{CZ} gate to perform a given two-qubit gate.

\subsection{Extended Temperley--Lieb configurations of the Yang--Baxter gate $B$ (\ref{braid matrix of BWM algebra})}

\label{subsection TL configuration of B}

In accordance with previous research \cite{ZZP15, ZZ14}, a given Yang--Baxter gate is allowed to have various of equivalent extended Temperley--Lieb
configurations. Here£¬ we investigate at least five types of extended Temperley--Lieb configurations of the Yang--Baxter gate $B$ (\ref{braid matrix of BWM algebra}),
three of which will be presented in this subsection and the remaining two of which will be presented in  Appendix \ref{appendix further study TL B}.

The formula (\ref{B gate as Bell like transform}) verifies that the Yang--Baxter gate $B$ (\ref{braid matrix of BWM algebra}) is the Bell transform (\ref{def_bell}) and
the associated Bell transform is expressed as
\eq
B=\sum_{k,l=0}^1e^{i\alpha_B}(R_\phi\otimes R_\phi H)|\psi(lk)\rangle\langle kl|.
\en
With the definition of the Bell basis  (\ref{definition four Bell states})  and the flow (\ref{relation flowing gate Bell state}) of a single-qubit gate on
 the Bell state $|\Psi\rangle$, the Yang--Baxter gate $B$ can be further reformulated as
\eq
\label{B gate expression Bell like transform}
B=\sum_{k,l=0}^1(1\!\!1_2\otimes V_{kl})|\Psi\rangle\langle kl|
\en
with $V_{kl}=e^{i\alpha_B}R_\phi H X^lZ^k R_\phi$. So the first extended Temperley--Lieb configuration of the Yang--Baxter gate $B$ is pictured as
  \eq
  \label{TL B Bell like transform}
  \setlength{\unitlength}{0.5mm}
  \begin{array}{c}
  \begin{picture}(100,40)

  \put(15,-2){\footnotesize{$B$}}

  \put(9,29){\line(0,1){6}}
  \put(9,11){\line(0,-1){6}}
  \put(27,29){\line(0,1){6}}
  \put(27,11){\line(0,-1){6}}

  \put(6,8){\line(0,1){24}}
  \put(30,8){\line(0,1){24}}
  \put(6,8){\line(1,0){24}}
  \put(6,32){\line(1,0){24}}

  \put(9,29){\line(1,-1){18}}
  \put(9,11){\line(1,1){7.3}}
  \put(27,29){\line(-1,-1){7.3}}

  \put(40,19){$=\sum_{k,l=0}^1$}

  \put(84,23){\line(0,1){12}}
  \put(74,23){\line(0,1){12}}
  \put(74,23){\line(1,0){10}}
  \put(84,29){\circle*{2.}}
  \put(86,28){\tiny{$V_{kl}$}}

  \put(74,5){\line(0,1){10}}
  \put(84,5){\line(0,1){10}}
  \put(84,11){\circle*{2.}}
  \put(85,10){\tiny{$X^l$}}
  \put(74,11){\circle*{2.}}
  \put(75,10){\tiny{$X^k$}}

  \put(71.9,15){\tiny{$\triangle$}}
  \put(81.9,15){\tiny{$\triangle$}}

  \end{picture}
  \end{array}
  \en
  where the vertical line with the symbol $\triangle$ represents the state $\langle 0|$ and such the line with the action of the Pauli gate $X$ represents the state $\langle 1|$.
  Note that this configuration is read from the bottom to the top, different from the convention of reading the quantum circuit (\ref{diagram quantum circuit decomposition B})  from
  the left to the right.

  Furthermore, the formula (\ref{B inverse gate as Bell like transform}) verifies that the inverse of the Yang--Baxter gate $B$ is also the Bell transform (\ref{def_bell}) and surprisingly after some algebra  the Yang--Baxter gate $B$ can be related to the inverse of Yang--Baxter gate $B$ with the local action of the single-qubit gate,
  \eq
  \label{B inverse gate expression as Bell like transform}
  B=\sum_{i,j=0}^1 |ij\rangle\langle\Psi|(1\!\!1_2\otimes U_{ij})
  \en
  with $U_{ij}=e^{-i\alpha_{B^\dag}}R_\phi^\dag Z^{i+1}X^{j+1}H R_\phi^\dag$. Thus the second extended Temperley--Lieb configuration of the Yang--Baxter gate $B$ is shown in
  \eq
  \label{TL B inverse Bell like transform}
  \setlength{\unitlength}{0.5mm}
  \begin{array}{c}
  \begin{picture}(100,40)

  \put(15,-2){\footnotesize{$B$}}

  \put(9,29){\line(0,1){6}}
  \put(9,11){\line(0,-1){6}}
  \put(27,29){\line(0,1){6}}
  \put(27,11){\line(0,-1){6}}

  \put(6,8){\line(0,1){24}}
  \put(30,8){\line(0,1){24}}
  \put(6,8){\line(1,0){24}}
  \put(6,32){\line(1,0){24}}

  \put(9,29){\line(1,-1){18}}
  \put(9,11){\line(1,1){7.3}}
  \put(27,29){\line(-1,-1){7.3}}

  \put(40,19){$=\sum_{i,j=0}^1$}

  \put(84,25){\line(0,1){10}}
  \put(74,25){\line(0,1){10}}
  \put(84,29){\circle*{2.}}
  \put(85,28){\tiny{$X^j$}}
  \put(74,29){\circle*{2.}}
  \put(75,28){\tiny{$X^i$}}

  \put(72.2,23){\tiny{$\nabla$}}
  \put(82.2,23){\tiny{$\nabla$}}

  \put(74,5){\line(0,1){12}}
  \put(84,5){\line(0,1){12}}
  \put(74,17){\line(1,0){10}}
  \put(83.9,11){\circle*{2.}}
  \put(85.9,10){\tiny{$U_{ij}$}}

  \end{picture}
  \end{array}
  \en
  which appears quite different from the configuration (\ref{TL B Bell like transform}), although they describe the same Yang--Baxter gate $B$ (\ref{braid matrix of BWM algebra}).

  Moreover, the third extended Temperley--Lieb configuration of the Yang--Baxter gate $B$ is due to the application of the spectral theorem \cite{NC2011}, that is,
  the Yang--Baxter gate $B$ has the decomposition
  \eq
  \label{relation B spectral theorem}
  B=\sum_{i,j=0}^1 \lambda_{ij}E_{ij},
  \en
  where the Temperley--Lieb projectors $E_{ij}$ (\ref{definition E ij projector}) are eigenstates of the Yang--Baxter gate $B$ with the respective eigenvalues
  $\lambda_{00}=e^{i5\pi/4}$, $\lambda_{01}=\lambda_{10}=e^{i3\pi/4}$ and $\lambda_{11}=e^{i\pi/4}$. The associated extended Temperley--Lieb configuration of
  the Yang--Baxter gate $B$ is shown as
  \eq
  \label{TL B spectral theorem}
\setlength{\unitlength}{0.5mm}
\begin{array}{c}
\begin{picture}(100,40)

\put(15,-2){\footnotesize{$B$}}

  \put(9,29){\line(0,1){6}}
  \put(9,11){\line(0,-1){6}}
  \put(27,29){\line(0,1){6}}
  \put(27,11){\line(0,-1){6}}

  \put(6,8){\line(0,1){24}}
  \put(30,8){\line(0,1){24}}
  \put(6,8){\line(1,0){24}}
  \put(6,32){\line(1,0){24}}

  \put(9,29){\line(1,-1){18}}
  \put(9,11){\line(1,1){7.3}}
  \put(27,29){\line(-1,-1){7.3}}

  \put(36,19){$=\sum_{i,j=0}^1\lambda_{ij}$}

  \put(88,23){\line(0,1){12}}
  \put(78,23){\line(0,1){12}}
  \put(78,23){\line(1,0){10}}
  \put(88,29){\circle*{2.}}
  \put(90,28){\tiny{$M_{ij}$}}

  \put(78,5){\line(0,1){12}}
  \put(88,5){\line(0,1){12}}
  \put(78,17){\line(1,0){10}}
  \put(87.9,11){\circle*{2.}}
  \put(89.9,10){\tiny{$M^\dag_{ij}$}}

\end{picture}
\end{array}
\en
with single-qubit gates $M_{ij}$ defined in (\ref{definition M ij gates}).

 Although three configurations (\ref{TL B Bell like transform}), (\ref{TL B inverse Bell like transform}) and (\ref{TL B spectral theorem}) of  the Yang--Baxter gate $B$  are equivalent, when and how they will be exploited completely depend on a specific circumstance; for example, the first and second configurations are immediately applied in the
 following subsection and the third configuration will be used in Section~\ref{section interpretation BMW teleportation}.

\subsection{Quantum teleportation using the Yang--Baxter gate $B$ (\ref{braid matrix of BWM algebra})}

The quantum teleportation \cite{BBCCJPW93, Vaidman94, BDM00, Werner01} can be characterized by the teleportation operator \cite{ZZP15,ZZ14}, which is the tensor product  in terms of the identity operator, the Bell transform and its inverse. Here the  Yang--Baxter gate $B$ is the Bell transform (\ref{def_bell}), see (\ref{B gate as Bell like transform}) and
(\ref{B gate expression Bell like transform}), and especially the inverse of the Yang--Baxter gate $B$ is related to the Yang--Baxter gate $B$ by the local
action of single-qubit gates, see (\ref{B inverse gate as Bell like transform}) and (\ref{B inverse gate expression as Bell like transform}). Hence we define the teleportation operator as the tensor product in terms of the identity operator and the  Yang--Baxter gate $B$, namely $(B\otimes 1\!\!1_2)(1\!\!1_2\otimes B)$ for simplicity (instead of $(B^{-1}\otimes 1\!\!1_2)(1\!\!1_2\otimes B)$ orginally introduced in \cite{ZZP15,ZZ14}).

With two types of the extended Temperley--Lieb configurations (\ref{TL B Bell like transform}) and (\ref{TL B inverse Bell like transform}) of the Yang--Baxter gate $B$,
the extended  Temperley--Lieb  configuration of the teleportation operator $(B\otimes 1\!\!1_2)(1\!\!1_2\otimes B)$ is illustrated in
  \eq
  \label{TL teleportation operator}
  \setlength{\unitlength}{0.5mm}
  \begin{array}{c}
  \begin{picture}(30,70)


  \put(-66,25){\line(0,1){6}}
  \put(-66,7){\line(0,-1){6}}
  \put(-48,25){\line(0,1){6}}
  \put(-48,7){\line(0,-1){6}}

  \put(-69,4){\line(0,1){24}}
  \put(-45,4){\line(0,1){24}}
  \put(-69,4){\line(1,0){24}}
  \put(-69,28){\line(1,0){24}}

  \put(-66,25){\line(1,-1){18}}
  \put(-66,7){\line(1,1){7.3}}
  \put(-48,25){\line(-1,-1){7.3}}


  \put(-84,55){\line(0,1){6}}
  \put(-84,37){\line(0,-1){6}}
  \put(-66,55){\line(0,1){6}}
  \put(-66,37){\line(0,-1){6}}

  \put(-87,34){\line(0,1){24}}
  \put(-63,34){\line(0,1){24}}
  \put(-87,34){\line(1,0){24}}
  \put(-87,58){\line(1,0){24}}

  \put(-84,55){\line(1,-1){18}}
  \put(-84,37){\line(1,1){7.3}}
  \put(-66,55){\line(-1,-1){7.3}}


  \put(-48,31){\line(0,1){30}}
  \put(-84,31){\line(0,-1){30}}


  \put(-40,29){$=$}


  \put(-30,29){$\sum_{i,j,k,l=0}^1$}

  \put(14.5,48){\line(0,1){12}}
  \put(2,48){\line(0,1){12}}
  \put(14.5,54){\circle*{2.}}
  \put(15.5,53){\tiny{$X^j$}}
  \put(2,54){\circle*{2.}}
  \put(3,53){\tiny{$X^i$}}

  \put(-.2,45){{\scriptsize$\nabla$}}
  \put(12.3,45){\scriptsize{$\nabla$}}

  \put(14.5,30){\circle*{2.}}
  \put(16,29){\tiny{$U_{ij}$}}

  \put(14.5,0){\line(0,1){12}}
  \put(27,0){\line(0,1){12}}
  \put(27,6){\circle*{2.}}
  \put(28,5){\tiny{$X^l$}}
  \put(14.5,6){\circle*{2.}}
  \put(15.5,5){\tiny{$X^k$}}

  \put(12.3,12){\tiny{$\triangle$}}
  \put(24.8,12){\tiny{$\triangle$}}

  \put(2,0){\line(0,1){37.5}}
  \put(2,37.5){\line(1,0){12.5}}
  \put(14.5,22.5){\line(0,1){15}}
  \put(27,22.5){\line(0,1){37.5}}
  \put(14.5,22.5){\line(1,0){12.5}}

  \put(27,30){\circle*{2.}}
  \put(28.5,29){\tiny{$V_{kl}$}}

  \put(40,29){$=\sum_{i,j,k,l=0}^1$}

  \put(94.5,48){\line(0,1){12}}
  \put(82,48){\line(0,1){12}}
  \put(94.5,54){\circle*{2.}}
  \put(95.5,53){\tiny{$X^j$}}
  \put(82,54){\circle*{2.}}
  \put(83,53){\tiny{$X^i$}}

  \put(79.8,45){{\scriptsize$\nabla$}}
  \put(92.3,45){\scriptsize{$\nabla$}}


  \put(94.5,0){\line(0,1){12}}
  \put(107,0){\line(0,1){12}}
  \put(107,6){\circle*{2.}}
  \put(108,5){\tiny{$X^l$}}
  \put(94.5,6){\circle*{2.}}
  \put(95.5,5){\tiny{$X^k$}}

  \put(92.3,12){\tiny{$\triangle$}}
  \put(104.8,12){\tiny{$\triangle$}}

  \put(82,0){\line(0,1){37.5}}
  \put(82,37.5){\line(1,0){12.5}}
  \put(94.5,22.5){\line(0,1){15}}
  \put(107,22.5){\line(0,1){37.5}}
  \put(94.5,22.5){\line(1,0){12.5}}

  \put(107,54){\circle*{2.}}
 \put(108.5,53){\tiny{$W_{i,j,k,l}$}}

  \end{picture}
  \end{array}
  \en
  where the single-qubit gate $U_{ij}$ (\ref{B inverse gate expression as Bell like transform}) flows from Alice's system to Bob's system  with the transpose and the single-qubit gate $W_{i,j,k,l}$ has the form
  \eq
  W_{i,j,k,l}=V_{kl}U^T_{ij}=(-1)^{l\cdot(k+j+1)}e^{i(\alpha_B-\alpha_{B^\dag})}R_\phi X^{j+k+1}Z^{i+l+1}R_\phi^\dag.
  \en
 In view of this configuration, it is obvious that an unknown qubit $|\alpha\rangle$ can be transmitted from Alice to Bob by topological deformation with the additional action
 of the single-qubit gate $W_{i,j,k,l}$. Consider the action of the teleportation operator $(B\otimes 1\!\!1_2)(1\!\!1_2\otimes B)$ on the initial state $|\alpha\rangle|kl\rangle$. With the help of the configuration (\ref{TL teleportation operator}),  the corresponding teleportation equation is expressed as
  \eq
  \label{teleportation equation B}
  (B\otimes 1\!\!1_2)(1\!\!1_2\otimes B)|\alpha\rangle|kl\rangle=\frac 1 2\sum_{i,j=0}^1|ij\rangle\otimes W_{i,j,k,l}|\alpha\rangle,
  \en
  where the factor $1/2$ is due to the vanishing of a pair of cup and cap configurations. To complete the teleportation protocol, Alice performs the product basis measurement
  $|ij\rangle\langle ij|$ on her qubits and then informs the measurement results $(i,j)$ to Bob. Afterwards, Bob applies the unitary correction gate $W^\dag_{i,j,k,l}$ on his qubit
   to obtain the transmitted state $|\alpha\rangle$.

   Two remarks are made. First, the fact that the teleportation operator $(B\otimes 1\!\!1_2)(1\!\!1_2\otimes B)$ successfully describes quantum teleportation means that the definition of the Bell transform (\ref{def_bell}) with additional local action of single-qubit gates is an appropriate generalization of the definition of the Bell transform in \cite{ZZ14}. Second, the notation $W_{i,j,k,l}$ (\ref{teleportation equation B}) is not the notation $W_{ij}$ in (\ref{definition four Bell states}).
   The single-qubit gate $W_{i,j,k,l}$ in the teleportation equation (\ref{teleportation equation B}) is in general not the Pauli gate because single-qubit gates $V_{kl}$ (\ref{B gate expression Bell like transform}) and $U_{ij}$ (\ref{B inverse gate expression as Bell like transform}) are usually not the Pauli gates, whereas the corresponding single-qubit gate in the teleportation equation derived in  \cite{ZZP15, ZZ14}  is always the Pauli gate. In this sense, quantum teleportation using the Yang--Baxter gate $B$ (\ref{braid matrix of BWM algebra}) is beyond the standard quantum teleportation in \cite{BBCCJPW93, Vaidman94, BDM00, Werner01} and  \cite{ZZP15, ZZ14}.

\section{Teleportation-based quantum computation using the Yang--Baxter gate}

\label{section teleportation quantum computation YBG}

In fault-tolerant quantum computation  \cite{NC2011, Preskill97}, Clifford gates \cite{NC2011, Gottesman97} such as the Hadamard gate and the phase
gate (\ref{definition H S R gates}) and the \textit{CZ} gate (\ref{definition CZ gate}), can be fault-tolerantly performed
in a systematic approach, but how a non-Clifford gate such as a specific phase shift gate $R_\phi$ (\ref{definition H S R gates}) with $\phi=\pi/8$
(called the $\pi/8$ gate or the $T$ gate, $T=R_{\pi/8}$), can be fault-tolerantly performed is not that explicit. Teleportation-based quantum
computation \cite{GC99, Nielsen03} is a type of fault-tolerant quantum computation in which the problem of how to fault-tolerantly
perform non-Clifford gates becomes the problem of how to fault-tolerantly perform Clifford gates with the aid of quantum teleportation.

In this section, first of all, we verify the Yang--Baxter gate $B_0$ (\ref{definition B 0 gate}) as a Clifford gate, which is the Yang--Baxter
gate $B$ (\ref{braid matrix of BWM algebra}) with the specific phase parameter $\phi=0$, and derive the teleportation equations using such the gate $B_0$.
Secondly, we study the fault-tolerant construction of single-qubit gates and two-qubit gates in teleportation-based quantum computation
using the Yang--Baxter  gate $B_0$ (\ref{definition B 0 gate}).

\subsection{The Yang--Baxter gate $B_0$ (\ref{definition B 0 gate}) is a Clifford gate}

A Clifford gate \cite{NC2011, Gottesman97} is expressed as a tensor product of the Hadamard gate $H$ and the phase gate $S$ (\ref{definition H S R gates})
and the \textit{CZ} gate (\ref{definition CZ gate}); equivalently, tensor products of the Pauli matrices with phase factors $\pm1$, $\pm i$
are preserved under conjugation of Clifford gates. Hence the Pauli gates with overall phases $\pm1$, $\pm i$ are the simplest examples for Clifford gates.
Note that Clifford gates play crucial roles in fault-tolerant quantum computation \cite{NC2011, Gottesman97}.

The Yang--Baxter gate $B$ (\ref{braid matrix of BWM algebra}) is in general not a Clifford gate  due to the phase shift gate $R_\phi$
(\ref{definition H S R gates}) in its decomposition formula (\ref{B decomposition}). When the phase factor is zero,  namely $\phi=0$, however,
the phase shift gate $R_\phi$ becomes the identity matrix and the associated Yang--Baxter gate denoted as $B_0$ given by
\eq
\label{definition B 0 gate}
B_0=\textit{CZ}(HZ\otimes HZ)\textit{CZ}
\en
with $Z=S^2$ modulo the overall phase $e^{i\frac 3 4\pi}$ is obviously a Clifford gate. The transformation properties of tensor products of the Pauli
matrices under conjugation by the Yang--Baxter gate $B_0$ shown below
\eq
\label{Clifford gate properties of B 0}
\begin{array}{ll} B_0 (X\otimes 1\!\!1_2) B_0^\dag = -(1\!\!1_2\otimes X), & B_0 (1\!\!1_2\otimes X) B_0^\dag =-(X\otimes 1\!\!1_2), \\
B_0 (Z\otimes 1\!\!1_2) B_0^\dag = (X\otimes Z), & B_0 (1\!\!1_2\otimes Z) B_0^\dag =(Z\otimes X), \end{array}
\en
also verify the Yang--Baxter gate $B_0$ as a Clifford gate.

The Yang--Baxter gate $B$ (\ref{braid matrix of BWM algebra}) and its inverse $B^\dag$ are the Bell transform (\ref{def_bell}), and thus the Yang--Baxter
gate $B_0$ (\ref{definition B 0 gate}) and its inverse $B_0^\dag$ are too. The $B_0$ gate acting on the product basis gives rise to the Bell basis with the
local action of the Hadamard gate,
\eq
B_0|ij\rangle=(-1)^{i+j+i\cdot j}(1\!\!1_2\otimes H)|\psi(ji)\rangle,
\en
which is the special case of (\ref{B gate as Bell like transform}) with $\phi=0$, and the same is true for the gate $B_0^\dag$,
\eq
B^\dag_0|ij\rangle=(-1)^{i\cdot j}(1\!\!1_2\otimes H)|\psi(j+1,i+1)\rangle,
\en
which is the special example of (\ref{B inverse gate as Bell like transform}) with $\phi=0$.

Consider the special case of the teleportation equation (\ref{teleportation equation B}) at  $\phi=0$. The associated teleportation operator
$(B_0\otimes 1\!\!1_2)(1\!\!1_2\otimes B_0)$ gives rise to the teleportation equation
\eq
\label{teleportation equation B 0}
(B_0\otimes 1\!\!1_2)(1\!\!1_2\otimes B_0)|\alpha\rangle|kl\rangle=\frac 1 2\sum_{i,j=0}^1|ij\rangle \otimes K_{i,j,k,l}|\alpha\rangle
\en
where single-qubit gates $K_{i,j,k,l}=(-1)^{j\cdot l+i\cdot j+k}X^{j+k+1}Z^{i+l+1}$ are the Pauli gates with phase factors. Furthermore, the teleportation operator
can be defined in the other way, namely  $(1\!\!1_2\otimes B_0)(B_0\otimes 1\!\!1_2)$ (refer to \cite{ZZP15}), and it is related to the teleportation equation
\eq
\label{teleportation equation B 0 inverse direction}
(1\!\!1_2\otimes B_0)(B_0\otimes 1\!\!1_2)|kl\rangle|\alpha\rangle=\frac 1 2\sum_{i,j=0}^1 L_{i,j,k,l}|\alpha\rangle\otimes|ij\rangle,
\en
where $L_{i,j,k,l}=(-1)^{i\cdot k+i\cdot j+l}X^{i+l+1}Z^{j+k+1}$. Both types of the teleportation equations using the Yang--Baxter gate $B_0$ are to be exploited in
the fault-tolerant construction of quantum gates in teleportation-based quantum computation in the following subsection.

\subsection{Fault-tolerant construction of a universal quantum gate set}

The detailed strategy of fault-tolerantly constructing single-qubit gates and two-qubit gates in teleportation-based quantum computation has been presented
in \cite{ZZP15}, and hence  we will make a brief sketch on the relevant  results in this paper.

In quantum circuit model \cite{NC2011}, single-qubit gates with an entangling two-qubit gate form a universal quantum gate set, namely all quantum gates
can be generated as a tensor product of gates from this gate set. All single-qubit gates can be generated from the Hadamard gate $H$ and the $T$ gate 
($T=R_{\pi/8}$) \cite{BMPRV00}. The Yang--Baxter  gate $B_0$ (\ref{definition B 0 gate}) is the Bell transform and so it is a maximal entangling two-qubit gate.
Therefore, the quantum gate set including the single-qubit gates $H$, $T$ and the  Yang--Baxter  gate $B_0$ is a universal quantum gate set.

To fault-tolerantly perform a single-qubit gate $U$ on an unknown qubit state $|\alpha\rangle$, namely $U|\alpha\rangle$,  we prepare the product 
state $|\alpha\rangle|kl\rangle$, then firstly apply $(1\!\!1_2\otimes B_0)$  and secondly apply $(1\!\!1_2\otimes 1\!\!1_2\otimes U)$, so that
the prepared state has the form,
\eq
(1\!\!1_2\otimes 1\!\!1_2\otimes U)(1\!\!1_2\otimes B_0)|\alpha\rangle|kl\rangle,
\en
where the single-qubit gate $U$ is acting on the entangled state $B_0|kl\rangle$. Thirdly, we apply $(B_0\otimes 1\!\!1_2)$ on such the prepared state 
in order to derive the teleportation equation given by
\eq
\label{teleportation equation single qubit gate}
(B_0\otimes 1\!\!1_2)(1\!\!1_2\otimes 1\!\!1_2\otimes U)(1\!\!1_2\otimes B_0)|\alpha\rangle|kl\rangle=\frac 1 2\sum_{i,j=0}^1|ij\rangle \otimes R_{i,j,k,l}U|\alpha\rangle,
\en
where $R_{i,j,k,l}=UK_{i,j,k,l}U^\dag$. Finally, we apply the local unitary correction operator $(1\!\!1_2\otimes 1\!\!1_2\otimes R^\dag_{i,j,k,l})$ to obtain the expected 
result $U|\alpha\rangle$. When the single-qubit gate $U$ is a Clifford gate such as the Hadamard gate $H$, the related single-qubit gates $R(H)_{i,j,k,l}$ have the form
\eq
R(H)_{i,j,k,l}=HK_{i,j,k,l}H^\dag=(-1)^{j\cdot l+i\cdot j+k}Z^{j+k+1}X^{i+l+1},
\en
which are Pauli gates with overall phases. When $U$ is a non-Clifford gate such as the $T$ gate,
the single-qubit gates $R(T)_{i,j,k,l}$ have the form
\eq
  R(T)_{i,j,k,l}=TK_{i,j,k,l}T^\dag=(-1)^{j\cdot l+i\cdot j+k}\left(\frac{X-\sqrt{-1}Y}{\sqrt 2}\right)^{j+k+1}Z^{i+l+1},
\en
with $\sqrt{-1}$ denoting the imaginary unit,  which are the Clifford gates. Therefore, the task of fault-tolerantly performing the non-Clifford gate $T$ in
teleportation-based quantum computation has changed as the task of fault-tolerantly preparing the initial state $(1\!\!1_2\otimes 1\!\!1_2\otimes T)(1\!\!1_2\otimes B_0)|\alpha\rangle|kl\rangle$ and  fault-tolerantly performing the Clifford gate $R(T)_{i,j,k,l}$.

To fault-tolerantly perform the Yang--Baxter gate $B_0$ on an unknown two-qubit state $|\alpha\beta\rangle$, namely $B_0|\alpha\beta\rangle$,
we fault-tolerantly prepare the initial state
\eq
(1\!\!1_2\otimes B_0\otimes 1\!\!1_2)(B_0\otimes B_0)(|k_1l_1\rangle\otimes|k_2l_2\rangle),
\en
and then perform a three-fold tensor product of the Yang--Baxter gate $B_0$ on such the prepared state, so that we have the teleportation equation
\eqa
&&(B_0\otimes B_0\otimes B_0)(1\!\!1_2\otimes B_0\otimes B_0\otimes 1\!\!1_2)(|\alpha\rangle\otimes|k_1l_1\rangle)\otimes(|k_2l_2\rangle\otimes|\beta\rangle) \nonumber \\
&=& \frac 1 4 \sum_{i_1,j_1=0}^1 \sum_{i_2,j_2=0}^1 (1\!\!1_4\otimes Q\otimes P\otimes 1\!\!1_4)(|i_1j_1\rangle\otimes B_0|\alpha\beta\rangle\otimes|i_2j_2\rangle),
\ena
where the tensor product of the single-qubit gates $Q$ and $P$ is defined as
\eq
Q\otimes P =B_0(K_{i_1,j_1,k_1,l_1}\otimes L_{i_2,j_2,k_2,l_2})B_0^\dag,
\en
 with the single-qubit gates $K_{i_1,j_1,k_1,l_1}$ and $L_{i_2,j_2,k_2,l_2}$ defined in the teleportation equations (\ref{teleportation equation B 0}) and (\ref{teleportation equation B 0 inverse direction}), respectively. After some algebra, the single-qubit gates $Q$ and $P$ are found to be the Pauli gates with phase factors,
\eqa
Q &=& (-1)^{(k_1+1)\cdot(i_1+l_1+1)+1}X^{i_1+i_2+l_1+l_2}Z^{j_2+k_2+1}; \\
P &=& (-1)^{i_2\cdot (k_2+j_2+1)+1}Z^{i_1+l_1+1}X^{j_1+k_1+j_2+k_2},
\ena
where the relations  (\ref{Clifford gate properties of B 0}) have been exploited.

\section{Interpretation of the tangle relations (\ref{definition BMW algebra_tangle}) of the BMW algebra via quantum teleportation}

\label{section interpretation BMW teleportation}

In the previous sections, we have reformulated the standard description of quantum teleportation using the Temperley--Lieb projector
and described both quantum teleportation protocol and teleportation-based quantum computation using the Yang--Baxter gate. On the other hand,
the mixed relations of the BMW algebra involving both the Temperley--Lieb projector and the Yang--Baxter gate, such as the tangle relations
 (\ref{definition BMW algebra_tangle}), have not been considered so far. In this section, therefore, we clearly show that the tangle relations
  of the BMW algebra have a natural interpretation of quantum teleportation.

First of all, from the tangle relations (\ref{definition BMW algebra_tangle}), we derive the constraint equations of single-qubit gates $M_{ij}$
(\ref{definition M ij gates}) which define both the  Temperley--Lieb projector $E_{ij}$ (\ref{definition E ij projector}) and the Yang--Baxter gate
$B$ (\ref{relation B spectral theorem}). For convenience, we make a theorem to summarize our research results and present a detailed proof for it.
Secondly, throughout the proof, we realize that the tangle relations of the BMW algebra are associated with the teleportation equations characterizing
quantum teleportation, which is summarized in the corollary to the theorem.

\begin{theorem}
\label{theorem1}
In terms of the Temperley--Lieb matrix $E$ (\ref{TL matrix of BMW algebra}) and the Yang--Baxter gate $B$ (\ref{braid matrix of BWM algebra}), the
four matrix expressions of the tangle relations (\ref{definition BMW algebra_tangle}) of the BMW algebra given by
\eqa
\label{definition four matrix tangle relation 1}
(B\otimes 1\!\!1_2)(1\!\!1_2\otimes B)(E\otimes 1\!\!1_2) &=& 2(1\!\!1_2\otimes E)(E\otimes 1\!\!1_2); \\
\label{definition four matrix tangle relation 2}
(1\!\!1_2\otimes B)(B\otimes 1\!\!1_2)(1\!\!1_2\otimes E) &=& 2(E\otimes 1\!\!1_2)(1\!\!1_2\otimes E); \\
\label{definition four matrix tangle relation 3}
(E\otimes 1\!\!1_2)(1\!\!1_2\otimes B)(B\otimes 1\!\!1_2) &=& 2(E\otimes 1\!\!1_2)(1\!\!1_2\otimes E); \\
\label{definition four matrix tangle relation 4}
(1\!\!1_2\otimes E)(B\otimes 1\!\!1_2)(1\!\!1_2\otimes B) &=& 2(1\!\!1_2\otimes E)(E\otimes 1\!\!1_2),
\ena
are respectively reduced to the constraint equations of  single-qubit gates $M_{ij}$ (\ref{definition M ij gates}) given by
   \eqa
   \label{relation tangle required relation 1}
   \sum_{k,l=0}^1\lambda_{ij}\lambda_{kl}M_{kl}M^*_{ij}M^T_{00}M^\dag_{kl} &=& 2M_{00}M^*_{ij}; \\
   \label{relation tangle required relation 2}
   \sum_{k,l=0}^1\lambda_{ij}\lambda_{kl}M^T_{kl}M^\dag_{ij}M_{00}M^*_{kl} &=& 2M^T_{00}M^\dag_{ij}; \\
   \label{relation tangle required relation 3}
   \sum_{k,l=0}^1\lambda_{ij}\lambda_{kl}M_{kl}M^*_{00}M^T_{ij}M^\dag_{kl} &=& 2M^T_{ij}M^\dag_{00}; \\
   \label{relation tangle required relation 4}
   \sum_{k,l=0}^1\lambda_{ij}\lambda_{kl}M^T_{kl}M^\dag_{00}M_{ij}M^*_{kl} &=& 2M_{ij}M^*_{00},
   \ena
where the parameters $\lambda_{ij}$ are the eigenvalues of the Yang--Baxter gate $B$ (\ref{relation B spectral theorem})
and single-qubit gates $M_{ij}$ (\ref{definition M ij gates}) are bases of unitary matrices defining the  Temperley--Lieb
projector $E_{ij}$ (\ref{definition E ij projector}).
\end{theorem}

\begin{proof}
As an example, let us derive the first constraint equation (\ref{relation tangle required relation 1})  of single-qubit gates  $M_{ij}$ from the first tangle relation
(\ref{definition four matrix tangle relation 1}). We apply both sides of the equation (\ref{definition four matrix tangle relation 1}) on $|\Psi_{M_{00}}\rangle\otimes |\alpha\rangle$
with the Bell-like state $|\Psi_{M_{00}}\rangle$ (\ref{definition M 00 state}) and an unknown qubit $|\alpha\rangle$, so that we have
\eq
\label{example tangle relation}
(B\otimes 1\!\!1_2)(1\!\!1_2\otimes B)\left(|\Psi_{M_{00}}\rangle\otimes |\alpha\rangle\right)= 2(1\!\!1_2\otimes E)(|\Psi_{M_{00}}\rangle\otimes |\alpha\rangle),
\en
where $E=|\Psi_{M_{00}}\rangle\langle\Psi_{M_{00}}|$ has been exploited.

With the teleportation equation (\ref{teleportation equation E 00 transpose}), the right-hand side of the equation (\ref{example tangle relation}) is equal to $|\alpha\rangle\otimes|\Psi_{M_{00}}\rangle$, which can be further reformulated  with the teleportation equation (\ref{teleportation equation M 00}), so that
we have
\eq
\label{relation tangle relation left side 1}
2(1\!\!1_2\otimes E)(|\Psi_{M_{00}}\rangle\otimes |\alpha\rangle)=\frac 1 2 \sum_{i,j=0}^1 |\Psi_{M_{ij}}\rangle\otimes M_{00}M^*_{ij}|\alpha\rangle.
\en
With the spectral decomposition of the Yang--Baxter gate $B$ (\ref{relation B spectral theorem}),  the left-hand side of the equation (\ref{example tangle relation})
can be calculated as follows,
\eqa
  \label{relation tangle relation right side 1}
  &&(B\otimes 1\!\!1_2)(1\!\!1_2\otimes B)(|\Psi_{M_{00}}\rangle\otimes |\alpha\rangle) \nonumber \\
   &=& \sum_{i,j=0}^1\sum_{k,l=0}^1\lambda_{ij}\lambda_{kl}\left((M_{ij},M^\dag_{ij})\otimes 1\!\!1_2\right)\left(1\!\!1_2\otimes (M_{kl},M^\dag_{kl})\right)\left(|\Psi_{M_{00}}\rangle\otimes |\alpha\rangle\right) \nonumber\\
   &=& \frac 1 4 \sum_{i,j=0}^1\sum_{k,l=0}^1 \lambda_{ij}\lambda_{kl}|\Psi_{M_{ij}}\rangle\otimes M_{kl}M^*_{ij}M^T_{00}M^\dag_{kl}|\alpha\rangle,
  \ena
where  the symbol $(M_{ij},M^\dag_{ij})$ denotes the Bell-like projective measurement operator,
\eq
  \label{definition notation projector with pair}
  (M_{ij},M^\dag_{ij})\equiv|\Psi_{M_{ij}}\rangle\langle\Psi_{M_{ij}}|.
  \en
Comparing both the equation (\ref{relation tangle relation left side 1}) and the equation (\ref{relation tangle relation right side 1}), we have
  \eq
  \label{nessary and sufficient}
  \frac 1 4 \sum_{i,j=0}^1\sum_{k,l=0}^1 \lambda_{ij}\lambda_{kl}|\Psi_{M_{ij}}\rangle\otimes M_{kl}M^*_{ij}M^T_{00}M^\dag_{kl}|\alpha\rangle=\frac 1 2 \sum_{i,j=0}^1 |\Psi_{M_{ij}}\rangle\otimes M_{00}M^*_{ij}|\alpha\rangle,
  \en
 which gives rise to the first constraint equation (\ref{relation tangle required relation 1}). Note that the  relation
 (\ref{relation tangle required relation 1}) is the necessary and sufficient condition for the relation (\ref{nessary and sufficient}) due to the facts
 that the Bell-like states $|\Psi_{M_{ij}}\rangle$ form an orthonormal basis of the two-qubit Hilbert space and the unknown state $|\alpha\rangle$ is arbitrary.

Similarly, the remaining tangle relations (\ref{definition four matrix tangle relation 2})-(\ref{definition four matrix tangle relation 4}) can be
replaced by
  \eqa
  \label{example tangle relation 2}
  (1\!\!1_2\otimes B)(B\otimes 1\!\!1_2)(|\alpha\rangle\otimes |\Psi_{M_{00}}\rangle) &=& 2(E\otimes 1\!\!1_2)(|\alpha\rangle\otimes |\Psi_{M_{00}}\rangle); \\
  \label{example tangle relation 3}
  (\langle\Psi_{M_{00}}|\otimes \langle\alpha|)(1\!\!1_2\otimes B)(B\otimes 1\!\!1_2) &=& 2(\langle\Psi_{M_{00}}|\otimes \langle\alpha|)(1\!\!1_2\otimes E); \\
  \label{example tangle relation 4}
  (\langle\alpha|\otimes \langle\Psi_{M_{00}}|)(B\otimes 1\!\!1_2)(1\!\!1_2\otimes B) &=& 2(\langle\alpha|\otimes \langle\Psi_{M_{00}}|)(E\otimes 1\!\!1_2),
  \ena
  which can be respectively used to derive the constraint relations (\ref{relation tangle required relation 2})-(\ref{relation tangle required relation 4}).
  \end{proof}

  Obviously, the above algebraic proof has a topological diagrammatical interpretation in the extended Temperley--Lieb diagrammatical approach.
  Note that the operator $(B\otimes 1\!\!1_2)(1\!\!1_2\otimes B)$ in the tangle relations (\ref{definition four matrix tangle relation 1})-(\ref{definition four matrix tangle relation 4}) does not play as the teleportation operator in the teleportation equation (\ref{teleportation equation B}) because it is not acting on the product state. Hence it is not appropriate to choose the extended Temperley--Lieb configurations  of the Yang--Baxter gate $B$ such as (\ref{TL B Bell like transform}) and (\ref{TL B inverse Bell like transform}) in the following discussion. Instead, we take account of the configuration (\ref{TL B spectral theorem}) of the Yang--Baxter gate $B$ to study the configuration of the tangle relations of the BMW algebra. The topological diagrammatical representation for the algebraic calculation in (\ref{relation tangle relation right side 1}) is illustrated in
  \eq
  \label{TL left hand side tangle relation}
  \setlength{\unitlength}{0.47mm}
  \begin{array}{c}
  \begin{picture}(78,78)


  \put(-86,40){\line(0,1){6}}
  \put(-86,22){\line(0,-1){23}}
  \put(-68,40){\line(0,1){6}}
  \put(-68,22){\line(0,-1){18.5}}

  \put(-89,19){\line(0,1){24}}
  \put(-65,19){\line(0,1){24}}
  \put(-89,19){\line(1,0){24}}
  \put(-89,43){\line(1,0){24}}

  \put(-86,40){\line(1,-1){18}}
  \put(-86,22){\line(1,1){7.5}}
  \put(-68,40){\line(-1,-1){7.5}}


  \put(-104,70){\line(0,1){6}}
  \put(-104,52){\line(0,-1){6}}
  \put(-86,70){\line(0,1){6}}
  \put(-86,52){\line(0,-1){6}}

  \put(-107,49){\line(0,1){24}}
  \put(-83,49){\line(0,1){24}}
  \put(-107,49){\line(1,0){24}}
  \put(-107,73){\line(1,0){24}}

  \put(-104,70){\line(1,-1){18}}
  \put(-104,52){\line(1,1){7.5}}
  \put(-86,70){\line(-1,-1){7.5}}


  \put(-68,46){\line(0,1){30}}
  \put(-104,46){\line(0,-1){47}}
  \put(-104,-1){\line(1,0){18}}

  \put(-86,8){\circle*{2.}}
  \put(-84,7){\tiny{$M_{00}$}}

  \put(-70.5,0){\footnotesize{$\nabla$}}



  \put(-60,35){$=\sum_{i,j,k,l=0}^1\lambda_{ij}\lambda_{kl}$}

  \put(14.5,60){\line(0,1){15}}
  \put(2,60){\line(0,1){15}}
  \put(2,60){\line(1,0){12.5}}
  \put(14.5,67.5){\circle*{2.}}
  \put(15.5,66.5){\tiny{$M_{ij}$}}

  \put(14.5,45){\circle*{2.}}
  \put(16,44){\tiny{$M^\dag_{ij}$}}

  \put(14.5,0){\line(0,1){30}}
  \put(27,3.5){\line(0,1){26.5}}
  \put(2,0){\line(1,0){12.5}}
  \put(14.5,30){\line(1,0){12.5}}
  \put(27,22.5){\circle*{2.}}
  \put(28,21.5){\tiny{$M^\dag_{kl}$}}

  \put(24.5,0){\footnotesize{$\nabla$}}

  \put(2,0){\line(0,1){52.5}}
  \put(2,52.5){\line(1,0){12.5}}
  \put(14.5,37.5){\line(0,1){15}}
  \put(27,37.5){\line(0,1){37.5}}
  \put(14.5,37.5){\line(1,0){12.5}}

  \put(27,45){\circle*{2.}}
  \put(28.5,44){\tiny{$M_{kl}$}}

  \put(14.5,8){\circle*{2.}}
  \put(15.5,7){\tiny{$M_{00}$}}

  \put(34,35){$=\frac 1 4\sum_{i,j,k,l=0}^1\lambda_{ij}\lambda_{kl}$}

  \put(115.5,30){\line(0,1){15}}
  \put(103,30){\line(0,1){15}}
  \put(103,30){\line(1,0){12.5}}
  \put(115.5,37.5){\circle*{2.}}
  \put(116.5,36.5){\tiny{$M_{ij}$}}

  \put(128,34){\line(0,1){11}}

  \put(128,37.5){\circle*{2.}}
 \put(129.5,36.5){\tiny{$M_{kl}M^*_{ij}M^T_{00}M^\dag_{kl}$}}

  \put(125.5,30){\footnotesize{$\nabla$}}

  \end{picture}
  \end{array}
  \en
  where the factor $\frac 1 4$ is due to the vanishing of two pairs of the cap and cup configurations in the topological straightening deformation and before
  such the deformation all relevant single-qubit gates have to be moved with the matrix transpose rule (\ref{relation flowing gate Bell state}) of flowing a single-qubit
  gate.

 Looking at both the algebraic and topological proofs for Theorem \ref{theorem1}, we see that each of the tangle relations of the BMW algebra is associated with the corresponding
 teleportation equation. This observation can be summarized in Corollary \ref{corollary1}  to Theorem \ref{theorem1}.
  \begin{corollary}
  \label{corollary1}
  The four matrix expressions of the tangle relations (\ref{definition four matrix tangle relation 1})-(\ref{definition four matrix tangle relation 4}) of the BMW algebra  are respectively associated with the teleportation equations,
  \eqa
   \label{teleportation equation M 00 1}
   |\alpha\rangle \otimes |\Psi_{M_{00}}\rangle &=& \frac 1 2 \sum_{i,j=0}^1  |\Psi_{M_{ij}}\rangle \otimes M_{00}M^*_{ij} |\alpha\rangle; \\
   \label{teleportation equation M 00 2}
   |\Psi_{M_{00}}\rangle \otimes |\alpha\rangle &=& \frac 1 2 \sum_{i,j=0}^1 M^T_{00}M^\dag_{ij} |\alpha\rangle\otimes |\Psi_{M_{ij}}\rangle; \\
   \label{teleportation equation M 00 3}
   \langle\alpha| \otimes \langle\Psi_{M_{00}}| &=& \frac 1 2 \sum_{i,j=0}^1  \langle\Psi_{M_{ij}}|\otimes\langle\alpha|M^T_{ij}M^\dag_{00}; \\
   \label{teleportation equation M 00 4}
   \langle\Psi_{M_{00}}| \otimes \langle\alpha| &=& \frac 1 2 \sum_{i,j=0}^1 \langle\alpha|M_{ij}M^*_{00} \otimes \langle\Psi_{M_{ij}}|.
   \ena
  \end{corollary}

  \begin{proof}
  For example, we draw Figure~\ref{fig_tangle relation and teleportation}  to understand the relationship between the tangle relation (\ref{definition four matrix tangle relation 1}) and the teleportation equation
  (\ref{teleportation equation M 00}) or (\ref{teleportation equation M 00 1}). In Figure~\ref{fig_tangle relation and teleportation}, the top two diagrams represent the tangle relation (\ref{definition four matrix tangle relation 1}) and the bottom two
   diagrams are derived from the top two diagrams, respectively, with the topological straightening deformation.  The bottom two diagrams are just both sides of the teleportation
   equation  (\ref{teleportation equation M 00 1}). Note that the constraint relation (\ref{relation tangle required relation 1}) has been already assumed in our study.
   Similarly, the other three tangle relations (\ref{definition four matrix tangle relation 2})-(\ref{definition four matrix tangle relation 4}) respectively lead to the
   teleportation equations (\ref{teleportation equation M 00 2})-(\ref{teleportation equation M 00 4}).
  \end{proof}

\begin{figure}
 \begin{center}
  \includegraphics[width=9cm]{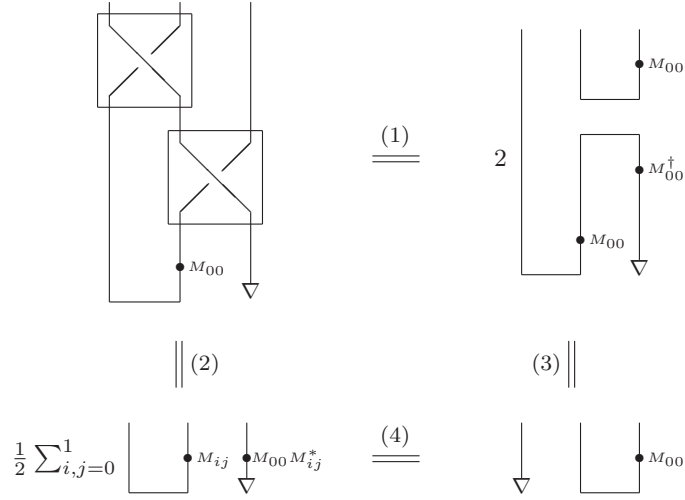}
  \end{center}
  \caption{\label{fig_tangle relation and teleportation}  Relationship between the tangle relation (\ref{definition four matrix tangle relation 1}) and the teleportation equation (\ref{teleportation equation M 00}) or (\ref{teleportation equation M 00 1}) in the topological diagrammatical representation. Let us explain the symbols (1), (2), (3) and (4), respectively. (1) The top two diagrams represent the tangle relation (\ref{example tangle relation}). (2) With the extended Temperley--Lieb configuration (\ref{TL B spectral theorem}) of the Yang--Baxter gate $B$ and the  constraint relation (\ref{relation tangle required relation 1}), the left-hand side of the tangle relation (\ref{definition four matrix tangle relation 1}) can be simplified into the right-hand side of the teleportation equation (\ref{teleportation equation M 00 1}),
  refer to  both (\ref{relation tangle relation right side 1})  and (\ref{TL left hand side tangle relation}). (3) The right-hand side of the tangle relation
   (\ref{definition four matrix tangle relation 1}) can be deformed into the left-hand side of the teleportation equation (\ref{teleportation equation M 00 1}) with the
   topological straightening operation. (4) The bottom two diagrams represent the teleportation equation (\ref{teleportation equation M 00 1}).
 }
\end{figure}

Therefore, in the extended Temperley--Lieb diagrammatical approach, the tangle relations (\ref{definition BMW algebra_tangle}) of the BMW algebra consisting
of both the Temperley--Lieb projector and the Yang--Baxter gate admit an interesting interpretation of quantum teleportation. Note that such the relationship
is independent of the specific representation of  the Temperley--Lieb projector and the Yang--Baxter gate, such as (\ref{TL matrix of BMW algebra})
and (\ref{braid matrix of BWM algebra}) presented in the reference \cite{WXSLZ15}. In view of this fact, furthermore, we will investigate the subject of how to
derive a general representation of the BMW algebra in the extended Temperley--Lieb diagrammatical approach in the next section.

\section{General construction of a representation of the tangle relations (\ref{definition BMW algebra_tangle}) of the BMW algebra}

\label{section construction rep from tangle relation}

It is well-known that how to construct an interesting braid representation (\ref{definition BMW algebra_YBE}) (or an interesting solution of the Yang--Baxter
equation \cite{YBE67}) is always a meaningful and challenge problem in the research frontier, refer to \cite{Kauffman02,YBE67, Dye03,KL04,ZKG04,ZZP15,PBR15}, and
the same is true for the construction of an interesting representation of the BMW algebra \cite{BW89}, refer to  \cite{Kauffman02,Jones89,CGX91, WXSZHW10, WXSLZ15,BMW11}.
In the previous research \cite{ZZP15}, a method of constructing the Yang--Baxter gate in the extended Temperley--Lieb diagrammatical approach has been proposed, and
hence in this section we want to exploit the extended Temperley--Lieb diagrammatical approach to construct interesting representations of the BMW algebra different 
from the representation using the Temperley--Lieb projector $E$ (\ref{TL matrix of BMW algebra}) and the Yang--Baxter gate $B$ (\ref{braid matrix of BWM algebra}). 
In view of the research result in the last section that the tangle relations (\ref{definition BMW algebra_tangle}) of the BMW algebra have a natural interpretation of 
quantum teleportation in the extended Temperley--Lieb diagrammatical approach, first of all, we focus on the general construction of the representation of the tangle 
relations (\ref{definition BMW algebra_tangle}) of the BMW algebra.

Looking at Theorem \ref{theorem1},  we see that the tangle relations  of the BMW algebra give rise to the constraint relations of single-qubit gates defining both
the Temperley--Lieb projector and  the Yang--Baxter gate. Hence we assume that the Temperley--Lieb projector $\tilde E$ and  the Yang--Baxter gate $U$ have the form in terms
of the Bell-like basis $|\Psi_{U_{ij}}\rangle$ given by
\eq
|\Psi_{U_{ij}}\rangle =(1\!\! 1_2 \otimes U_{ij} ) |\Psi\rangle
\en
with single-qubit gates $U_{ij}$ satisfying the orthonormal condition
\eq
\label{orthonormal condition U}
\frac 1 2 \tr\left(U^\dag_{i_2j_2}U_{i_1j_1}\right)=\delta_{i_1i_2}\delta_{j_1j_2},
\en
and then require them to  satisfy the tangle relations of the BMW algebra in order to derive the constraint equations of single-qubit gates $U_{ij}$.

\begin{theorem}
\label{theorem2}
The Temperley--Lieb projector $\tilde E$ assumes the form $\tilde E=|\Psi_{U_{mn}}\rangle\langle\Psi_{U_{mn}}|$ with specified $m, n$ and
the Yang--Baxter gate $U$ assumes the form of the spectral decomposition given by
\eq
\label{relation spectral theorem U gate}
U=\sum_{i,j=0}^1\mu_{ij}|\Psi_{U_{ij}}\rangle\langle\Psi_{U_{ij}}|,
\en
with  eigenvalues $\mu_{ij}$. When single-qubit gates $U_{ij}$, $i,j=0,1$ and the specified single-qubit gate $U_{mn}$
satisfy the constraint relations:
\eqa
\label{relation construction tangle relation constraint 1}\frac 1 2 \sum_{k,l=0}^1\mu_{ij}\mu_{kl}U^\dag_{mn}U_{kl}U_{ij}^*U_{mn}^TU^\dag_{kl}U_{mn} &=& U^*_{ij}U^T_{mn};\\
\label{relation construction tangle relation constraint 2}\frac 1 2\sum_{k,l=0}^1\mu_{ij}\mu_{kl}U_{mn}^*U^T_{kl}U_{ij}^\dag U_{mn}U^*_{kl}U^T_{mn} &=& U^\dag_{ij}U_{mn};\\
\label{relation construction tangle relation constraint 3}\frac 1 2 \sum_{k,l=0}^1\mu_{ij}\mu_{kl}U^\dag_{mn}U_{kl}U_{mn}^*U_{ij}^TU^\dag_{kl}U_{mn} &=& U^*_{mn}U^T_{ij};\\
\label{relation construction tangle relation constraint 4}\frac 1 2 \sum_{k,l=0}^1\mu_{ij}\mu_{kl}U_{mn}^* U^T_{kl}U^\dag_{mn}U_{ij}U^*_{kl}U^T_{mn} &=& U^\dag_{mn}U_{ij},
\ena
the Temperley--Lieb projector $\tilde E$ and the Yang--Baxter gate $U $ satisfy the tangle relations (\ref{definition BMW algebra_tangle}) of the BMW algebra.
\end{theorem}

\begin{proof}
The proof is a direct generalization of the proof for Theorem~\ref{theorem1}. For example, we derive the constraint 
relation (\ref{relation construction tangle relation constraint 1})
from the tangle relation using the  Temperley--Lieb projector $\tilde E$ and the Yang--Baxter gate $U$ given by
\eq
\label{relation construction tangle relation 1}(U\otimes 1\!\!1_2)(1\!\!1_2\otimes U)(\tilde E\otimes 1\!\!1_2) = 2(1\!\!1_2\otimes \tilde E)(\tilde E\otimes 1\!\!1_2).
\en
Both sides of the equation (\ref{relation construction tangle relation 1}) acting on $|\Psi_{M_{00}}\rangle\otimes |\alpha\rangle$ give rise to
\eqa
&&\sum_{i,j=0}^1\sum_{k,l=0}^1\mu_{ij}\mu_{kl}\left((U_{ij},U^\dag_{ij}),1\!\!1_2\right)
\left(1\!\!1_2,(U_{kl},U^\dag_{kl})\right)\left(|\Psi_{U_{mn}}\rangle\otimes|\alpha\rangle\right) \nonumber \\
&=&2\left(1\!\!1_2,(U_{mn},U^\dag_{mn})\right)\left(|\Psi_{U_{mn}}\rangle\otimes|\alpha\rangle\right),
\ena
with the Bell-like projector $(U_{ij},U^\dag_{ij})$ defined in (\ref{definition notation projector with pair}), which can be simplified with the help of extended
Temperley--Lieb  configurations such as (\ref{TL teleportation transposed M 00}) and (\ref{TL left hand side tangle relation}),
\eq
\label{relation construction tangle reduece relation}
\frac 1 4\sum_{i,j=0}^1\sum_{k,l=0}^1|\Psi_{U_{ij}}\rangle\otimes \mu_{ij}\mu_{kl}U_{kl}U_{ij}^*U_{mn}^TU^\dag_{kl}|\alpha\rangle=U^T_{mn}U^\dag_{mn}|\alpha\rangle\otimes|\Psi_{U_{mn}}\rangle.
\en
Furthermore, with the teleportation equation of describing the transport of the unknown qubit $|\beta\rangle= U^T_{mn}U^\dag_{mn}|\alpha\rangle$,
\eq
\label{relation construction tangle relaton tele equation}
|\beta\rangle\otimes|\Psi_{U_{mn}}\rangle=\frac 1 2\sum_{i,j=0}^1|\Psi_{U_{ij}}\rangle\otimes U_{mn}U^*_{ij}|\beta\rangle,
\en
which can be derived with the extended Temperley--Lieb configuration (\ref{TL teleportation M 00}), the equation (\ref{relation construction tangle reduece relation})
leads to the constraint relation of single-qubit gates $U_{ij}$,
\eq
\frac 1 2 \sum_{k,l=0}^1\mu_{ij}\mu_{kl}U_{kl}U_{ij}^*U_{mn}^TU^\dag_{kl}U_{mn}U^*_{mn}=U_{mn}U^*_{ij},
\en
which is equivalent to (\ref{relation construction tangle relation constraint 1}).
\end{proof}

\begin{corollary}
The eigenvalues $\mu_{ij}$ of the Yang--Baxter gate $U$ (\ref{relation spectral theorem U gate}) satisfy the constraint relation
\eq
\frac 1 2 \sum_{k,l=0}^1\mu_{mn}\mu_{kl}=1,
\en
when the Temperley--Lieb projector $\tilde E$ and the Yang--Baxter gate $U $ satisfy the tangle relations (\ref{definition BMW algebra_tangle}) of the BMW algebra.
\end{corollary}

\begin{proof}
The trace of the constraint relation (\ref{relation construction tangle relation constraint 1}) has the form
\eq
\frac 1 2 \sum_{k,l=0}^1\mu_{ij}\mu_{kl}\tr\left(U^\dag_{mn}U_{kl}U_{ij}^*U_{mn}^TU^\dag_{kl}U_{mn}\right)=\tr\left(U^*_{ij}U^T_{mn}\right)
\en
which leads to
\eq
\label{relation trace tangle required relation}
\frac 1 2 \sum_{k,l=0}^1\mu_{ij}\mu_{kl}\delta_{im}\delta_{jn}=\delta_{im}\delta_{jn},
\en
where the orthonormal condition (\ref{orthonormal condition U}) for the Bell-like basis states $|\Psi_{U_{ij}}\rangle$ has been applied. Both sides of
the constraint relation (\ref{relation trace tangle required relation}) have the summation over the indices $i$, $j$, which completes the proof for
the corollary.
\end{proof}

As an example, we solve the constraint relations (\ref{relation construction tangle relation constraint 1})-(\ref{relation construction tangle relation constraint 4})
to obtain interesting representations of the BMW algebra. The unitary bases $U_{ij}$ are set as $W_{ij}=X^i Z^j$ defining the Bell states
$|\psi(ij)\rangle$ (\ref{definition four Bell states}), and the unitary matrix $U_{mn}$ defining the Temperley--Lieb projector $\tilde E$  is also chosen
as $U_{mn}=X^m Z^n$. The detailed calculation is shown in Appendix \ref{appendix trial solution tangle relation}, and the results are summarized as follows.
\begin{itemize}
  \item The Temperley--Lieb projector $\tilde E$ has the form
  \eq
  \label{representation tangle relation E 1}
  \tilde E=\frac 1 2 \left(
                       \begin{array}{cccc}
                         1 & 0 & 0 & \epsilon \\
                         0 & 0 & 0 & 0 \\
                         0 & 0 & 0 & 0 \\
                         \epsilon & 0 & 0 & 1 \\
                       \end{array}
                     \right),
\en
where $\epsilon=1$ for $\tilde E=|\psi(00)\rangle \langle \psi(00)|$ and
  $\epsilon=-1$ for $\tilde E=|\psi(01)\rangle \langle \psi(01)|$,
and the associated Yang--Baxter gate $U$ has the form
  \eq
  \label{representation tangle relation U 1}
U=\left(
         \begin{array}{cccc}
         \cos\phi & 0 & 0 & i\sin\phi \\
         0 & -i\epsilon\sin\phi & \pm\cos\phi & 0 \\
         0 & \pm\cos\phi & -i\epsilon\sin\phi & 0 \\
         i\sin\phi & 0 & 0 & \cos\phi \\
         \end{array}
         \right),
  \en
  or
  \eq
  \label{representation tangle relation U 2}
  U=\left(
                                                                \begin{array}{cccc}
                                                                  0 & 0 & 0 & e^{i\phi} \\
                                                                  0 & \epsilon e^{-i\phi} & 0 & 0 \\
                                                                  0 & 0 & \epsilon e^{-i\phi} & 0 \\
                                                                  e^{i\phi} & 0 & 0 & 0 \\
                                                                \end{array}
                                                              \right),
  \en
  with $\phi\in[0,2\pi)$.
  \item The Temperley--Lieb projector $\tilde E$ has the form
  \eq
  \label{representation tangle relation E 2}
  \tilde E=\frac 1 2 \left(
                       \begin{array}{cccc}
                         0 & 0 & 0 & 0 \\
                         0 & 1 & \epsilon & 0 \\
                         0 & \epsilon & 1 & 0 \\
                         0 & 0 & 0 & 0 \\
                       \end{array}
                     \right),
  \en
  where $\epsilon=1$ for $\tilde E=|\psi(10)\rangle \langle \psi(10)|$ and
  $\epsilon=-1$ for $\tilde E=|\psi(11)\rangle \langle \psi(11)|$,
and the corresponding Yang--Baxter gate $U$ has the form
  \eq
  \label{representation tangle relation U 3}
  U=\left(
    \begin{array}{cccc}
                                                                  i\sin\phi & 0 & 0 & \cos\phi \\
                                                                  0 & \pm\cos\phi & -i\epsilon\sin\phi & 0 \\
                                                                  0 & -i\epsilon\sin\phi & \pm\cos\phi & 0 \\
                                                                  \cos\phi & 0 & 0 & i\sin\phi \\
                                                                \end{array}
    \right),
  \en
  or
  \eq
  \label{representation tangle relation U 4}
  U=\left(
                                                                \begin{array}{cccc}
                                                                  e^{i\phi} & 0 & 0 & 0 \\
                                                                  0 & 0 & \epsilon e^{-i\phi} & 0 \\
                                                                  0 & \epsilon e^{-i\phi} & 0 & 0 \\
                                                                  0 & 0 & 0 & e^{i\phi} \\
                                                                \end{array}
                                                              \right),
  \en
  with $\phi\in[0,2\pi)$.
\end{itemize}

In the viewpoint of the extended Temperley--Lieb diagrammatical approach \cite{ZZP15}, Theorem \ref{theorem2} can be easily generalized. We consider
the general case for the construction of the representation of the tangle relations (\ref{definition BMW algebra_tangle}) in which a two-qubit projective
measurement operator $\tilde E$ and a two-qubit gate $G$ are involved, namely that the Temperley--Lieb projector $\tilde E$ and the Yang--Baxter gate $G$ are
not supposed.

\begin{theorem}
\label{theorem3}
A two-qubit projector measurement operator $\tilde E=|\Psi_{U_{mn}}\rangle\langle\Psi_{U_{mn}}|$
and a two-qubit quantum gate $G$ given by
\eq
G=\sum_{i,j,k,l=0}^1\tilde G_{ij,kl}|\Psi_{U_{ij}}\rangle\langle\Psi_{U_{kl}}|
\en
with 16 entries of complex numbers $\tilde  G_{ij,kl}$ form the representation of the tangle relations (\ref{definition BMW algebra_tangle})
when the constraint relations
\eqa
\label{relation construction general tangle relation constraint 1}
\frac 1 2\sum_{k_1,l_1=0}^1\sum_{i_2,j_2,k_2,l_2=0}^1\tilde G_{i_1j_1,k_1l_1}\tilde G_{i_2j_2,k_2l_2}U^\dag_{mn}U_{i_2j_2}U^*_{k_1l_1}U^T_{mn}U_{k_2l_2}^\dag U_{mn}=& U^*_{i_1j_1}U^T_{mn};\\
\label{relation construction general tangle relation constraint 2}
\frac 1 2\sum_{k_1,l_1=0}^1\sum_{i_2,j_2,k_2,l_2=0}^1\tilde G_{i_1j_1,k_1l_1}\tilde G_{i_2j_2,k_2l_2}U^*_{mn}U^T_{i_2j_2}U^\dag_{k_1l_1}U_{mn}U^*_{k_2l_2} U^T_{mn}=& U^\dag_{i_1j_1}U_{mn};\\
\label{relation construction general tangle relation constraint 3}
\frac 1 2\sum_{k_1,l_1=0}^1\sum_{i_2,j_2,k_2,l_2=0}^1\tilde G_{i_1j_1,k_1l_1}\tilde G_{i_2j_2,k_2l_2}U^\dag_{mn}U_{i_2j_2}U^*_{mn}U^T_{k_1l_1}U_{k_2l_2}^\dag U_{mn}=& U^*_{mn}U^T_{i_1j_1};\\
\label{relation construction general tangle relation constraint 4}
\frac 1 2\sum_{k_1,l_1=0}^1\sum_{i_2,j_2,k_2,l_2=0}^1\tilde G_{i_1j_1,k_1l_1}\tilde G_{i_2j_2,k_2l_2}U^*_{mn}U^T_{i_2j_2}U^\dag_{mn}U_{k_1l_1}U^*_{k_2l_2} U^T_{mn}=& U^\dag_{mn}U_{i_1j_1},
\ena
are satisfied.
\end{theorem}

The proof for Theorem \ref{theorem3} is a direct generalization of the proof for Theorem \ref{theorem2}, and so it is omitted here. About how to solve these constraint
relations to obtain an interesting representation of the tangle relations (\ref{definition BMW algebra_tangle}) is a challenge problem in future research, because the
result may be not a representation of the BMW algebra but indeed has an interesting interpretation of quantum teleportation. In addition, the constraint relations
(\ref{relation construction tangle relation constraint 1})-(\ref{relation construction tangle relation constraint 4}) and the constraint relations (\ref{relation construction general tangle relation constraint 1})-(\ref{relation construction general tangle relation constraint 4}) can be reformulated with the new conventions and notations, refer to Appendix \ref{skew transpose}.

\section{Concluding remarks}

\label{section concluding remarks}

In this paper, we describe quantum teleportation protocol  \cite{BBCCJPW93, Vaidman94, BDM00, Werner01} and teleportation-based quantum computation \cite{GC99,Nielsen03} using the generators of the BMW algebra including both the Yang--Baxter gate and the Temperley--Lieb projector. We point out that the tangle relations defining the BMW algebra have a close connection with the teleportation process, and thus the extended Temperley--Lieb diagrammatical approach \cite{Zhang06,ZZP15} properly characterizes the topological feature of quantum teleportation. We propose a meaningful approach of constructing a general representation of the tangle relations of the BMW algebra and obtain interesting representations
of the BMW algebra.

{\em Notes Added}. After this paper is done, we are occasionally informed that the Yang--Baxter gates (\ref{representation tangle relation U 1}) and
(\ref{representation tangle relation U 3}) have been already presented in the preprint \cite{PBR15}. As a matter of fact, these gates are derived in two essentially
different approaches. We derive such the Yang--Baxter gates in the extended Temperley--Lieb diagrammatical approach \cite{Zhang06,ZZP15}, whereas the authors of \cite{PBR15}
obtain them via the algebraic approach of the cyclic group. We study and look for interesting representations of the BMW algebra, which are not involved  in \cite{PBR15}.

\section*{Acknowledgements}

This work was supported by the starting Grant 273732 of Wuhan University, P. R. China and is supported by
the NSF of China (Grant No. 11574237 and 11547310).

\appendix

\section{The Brauer algebra and quantum teleportation}

\label{appendix Brauer algebra}

It is well-known that the BMW algebra \cite{BW89} is the algebraic deformation of the Brauer algebra \cite{Brauer37}. Quantum teleportation using
the Brauer algebra has been explored in \cite{Zhang06}, which originally motivated the authors to study quantum teleportation using the BMW
algebra and write down the present paper. Here we make a simple sketch on the Brauer algebra and its relation to quantum teleportation.

The Brauer algebra $D_n(d)$ \cite{Brauer37} with the loop parameter $d$ is generated by the Temperley--Lieb idempotents $e_i$ and the
permutations $v_i$ with $i=1,\ldots, n-1$. The Temperley--Lieb idempotents $e_i$ satisfy the algebraic relations,
\eq
\label{definition Brauer algebra}
  \begin{array}{ll}
    e_{i}^2=e_{i} &  e_{i}e_{i\pm 1}e_{i}=d^{-2}e_{i}, \\
    e_ie_j=e_je_i & |i-j|\geq 2,
  \end{array}
\en
and the permutation generators $v_i$ satisfy
\eq
    v_i^2=1\!\!1,\quad
    v_{i}v_{i\pm 1}v_{i}=v_{i\pm 1}v_{i}v_{i\pm 1}.
\en
Both generators satisfy the first type of the mixed relations,
\eq
  e_{i}v_i=v_{i}e_{i}= e_i,
\en
and the second type of the mixed relations,
\eq
\label{tangle_Brauer}
v_{i\pm 1}v_{i}e_{i\pm 1}=e_{i}v_{i\pm 1}v_{i}=de_{i}e_{i\pm 1},
\en
which are called the tangle relations of the Brauer algebra in this paper.

A tensor product representation of the Brauer algebra can be constructed in terms of the Bell state
projector $|\Psi\rangle\langle\Psi|$ (\ref{definition four Bell states}) and the permutation gate $P$
defined by $P|ij\rangle=|ji\rangle$ as follows
\eqa
e_i &=& 1\!\!1^{\otimes(i-1)}\otimes |\Psi\rangle\langle\Psi|\otimes 1\!\!1^{\otimes(n-i-1)},\\
v_i &=& 1\!\!1^{\otimes(i-1)}\otimes P\otimes 1\!\!1^{\otimes(n-i-1)},
\ena
with the loop parameter $d=2$. Using the Bell state projector  $|\Psi\rangle\langle\Psi|$, we perform the teleportation process
in the way
\eq
 (|\Psi\rangle\langle\Psi|\otimes 1\!\! 1_2)(|\alpha\rangle\otimes |\Psi\rangle) = \frac 1 2 |\Psi\rangle \otimes  |\alpha\rangle,
 \en
which is a special case of (\ref{teleportation equation projector})  for $i=j=0$. In terms of the permutation gate $P$, we define
the teleportation operator as $(1\!\!1_2\otimes P)(P\otimes 1\!\!1_2)$ to swap the quantum state in the way
\eq
 (1\!\!1_2\otimes P)(P\otimes 1\!\!1_2)(|\alpha\rangle\otimes |ij\rangle) = |ij\rangle\otimes|\alpha\rangle.
 \en
With the configuration (\ref{TL Bell measurement}) of the Bell state projector $|\Psi\rangle\langle\Psi|$
and the extended Temperley--Lieb configuration  of the permutation gate $P$,
 \eq
  \setlength{\unitlength}{0.5mm}
  \begin{array}{c}
  \begin{picture}(130,32)

  \put(20,16){$P=\sum_{i,j=0}^1(-1)^{i\cdot j}$}

  \put(92,18){\line(0,1){12}}
  \put(82,18){\line(0,1){12}}
  \put(82,18){\line(1,0){10}}
  \put(92,24){\circle*{2.}}
  \put(94,23){\tiny{$W_{ij}$}}
  \put(82,0){\line(0,1){12}}
  \put(92,0){\line(0,1){12}}
  \put(82,12){\line(1,0){10}}
  \put(92,6){\circle*{2.}}
  \put(94,5){\tiny{$W^\dag_{ij}$}}
 \put(110,12){,}
  \end{picture}
  \end{array}
  \en
we reformulate the tangle relations  (\ref{tangle_Brauer})  of the Brauer algebra as
 \eq
  (P\otimes 1\!\!1_2)(1\!\!1_2\otimes P)\left(|\Psi\rangle\otimes |\alpha\rangle\right)
  = 2(1\!\!1_2\otimes |\Psi\rangle\langle\Psi|)\left(|\Psi\rangle\otimes |\alpha\rangle\right)
  \en
 which can be reduced into the teleportation equation (\ref{teleportation equation original}), refer to the proof for Corollary~\ref{corollary1}. 

When the permutation gate $P$ is replaced with the Yang--Baxter gate $B$ (\ref{braid matrix of BWM algebra}),
the above representation of the Brauer algebra is substituted by  the representation of the BMW algebra,
so the study of quantum teleportation using the Brauer algebra in \cite{Zhang06} naturally points towards the study
of quantum teleportation using the BMW algebra in this paper.

\section{More on the extended Temperley--Lieb configurations of the Yang--Baxter gate $B$ (\ref{braid matrix of BWM algebra})}

\label{appendix further study TL B}

It is obvious that the extended Temperley--Lieb configurations \cite{Zhang06, ZZP15} play the essential roles throughout this paper. In accordance with the
reference \cite{ZZP15}, a Yang--Baxter gate allows various but equivalent extended Temperley--Lieb configurations. For example, three
distinct configurations of the  Yang--Baxter gate $B$ (\ref{braid matrix of BWM algebra}) are presented in Subsection~\ref{subsection TL configuration of B}.
Here another two extended Temperley--Lieb configurations of the Yang--Baxter gate $B$ are introduced and they may be useful elsewhere.

The Yang--Baxter gate $B$ (\ref{braid matrix of BWM algebra}) can be related to the Temperley--Lieb projector $E$ (\ref{TL matrix of BMW algebra}) in the way
\eq
\label{B gate decomposition in E}
B=\tilde U+2ie^{i\frac 3 4 \pi}E,
\en
where $\tilde U$ is a unitary matrix with the decomposition
\eq
\tilde U=e^{i\frac 3 4\pi}1\!\!1_4+\sqrt2\left(|\psi(10)\rangle\langle\psi(10)|+|\Psi_{R_{2\phi}}\rangle\langle\Psi_{R_{2\phi}}|\right),
\en
with $|\psi(10)\rangle=(1\!\!1_2\otimes X)|\Psi\rangle$ and $|\Psi_{R_{2\phi}}\rangle=(1\!\!1_2\otimes R_{2\phi})|\Psi\rangle$, $R_{2\phi}$ denoting
the phase shift gate (\ref{definition H S R gates}), and the associated extended Temperley--Lieb configuration is illustrated in
\eq
\label{TL B other form 1}
  \setlength{\unitlength}{0.5mm}
  \begin{array}{c}
  \begin{picture}(130,40)

  \put(-25,-2){\footnotesize{$B$}}

  \put(-31,29){\line(0,1){6}}
  \put(-31,11){\line(0,-1){6}}
  \put(-13,29){\line(0,1){6}}
  \put(-13,11){\line(0,-1){6}}

  \put(-34,8){\line(0,1){24}}
  \put(-10,8){\line(0,1){24}}
  \put(-34,8){\line(1,0){24}}
  \put(-34,32){\line(1,0){24}}

  \put(-31,29){\line(1,-1){18}}
  \put(-31,11){\line(1,1){7.3}}
  \put(-13,29){\line(-1,-1){7.3}}

  \put(-3,19){$=e^{i\frac 3 4\pi}$}

  \put(32,5){\line(0,1){30}}
  \put(22,5){\line(0,1){30}}

  \put(40,19){$+\sqrt2$}

  \put(72,23){\line(0,1){12}}
  \put(62,23){\line(0,1){12}}
  \put(62,23){\line(1,0){10}}
  \put(72,29){\circle*{2.}}
  \put(74,28){\tiny{$X$}}

  \put(62,5){\line(0,1){12}}
  \put(72,5){\line(0,1){12}}
  \put(62,17){\line(1,0){10}}
  \put(72,11){\circle*{2.}}
  \put(74,10){\tiny{$X$}}

    \put(80,19){$+\sqrt2$}

  \put(112,23){\line(0,1){12}}
  \put(102,23){\line(0,1){12}}
  \put(102,23){\line(1,0){10}}
  \put(112,29){\circle*{2.}}
  \put(114,28){\tiny{$R_{2\phi}$}}

  \put(102,5){\line(0,1){12}}
  \put(112,5){\line(0,1){12}}
  \put(102,17){\line(1,0){10}}
  \put(112,11){\circle*{2.}}
  \put(114,10){\tiny{$R^\dag_{2\phi}$}}

      \put(118,19){$+2ie^{i\frac 3 4 \pi}$}

  \put(156,23){\line(0,1){12}}
  \put(146,23){\line(0,1){12}}
  \put(146,23){\line(1,0){10}}
  \put(156,29){\circle*{2.}}
  \put(158,28){\tiny{$M_{00}$}}

  \put(146,5){\line(0,1){12}}
  \put(156,5){\line(0,1){12}}
  \put(146,17){\line(1,0){10}}
  \put(156,11){\circle*{2.}}
  \put(158,10){\tiny{$M_{00}^\dag$}}

  \end{picture}
  \end{array}
  \en
  where the two vertical lines represent for the identity matrix $1\!\!1_4$.

  Note that the single-qubit gate $M_{00}$ is defined as $M_{00}=R_\phi HSH R_\phi$ (\ref{definition M ij gates}).
  And bring such decomposition of the $M_{00}$ gate into the relation (\ref{B gate decomposition in E}). After some
  algebra, the Yang--Baxter gate $B$ takes another form
  \eq
  B=e^{i\frac 3 4\pi}\left(1\!\!1_4-|\psi(10)\rangle\langle\psi(10)|-|\Psi_{R_{2\phi}}\rangle\langle\Psi_{R_{2\phi}}|
  -e^{-i\phi}|\Psi_{R_{2\phi}}\rangle\langle\psi(10)|+e^{i\phi}|\psi(10)\rangle\langle\Psi_{R_{2\phi}}|\right),
  \en
  and its extended Temperley--Lieb configuration is shown below
  \eq
  \label{TL B other form 2}
  \setlength{\unitlength}{0.5mm}
  \begin{array}{c}
  \begin{picture}(135,40)

  \put(-25,-2){\footnotesize{$B$}}

  \put(-31,29){\line(0,1){6}}
  \put(-31,11){\line(0,-1){6}}
  \put(-13,29){\line(0,1){6}}
  \put(-13,11){\line(0,-1){6}}

  \put(-34,8){\line(0,1){24}}
  \put(-10,8){\line(0,1){24}}
  \put(-34,8){\line(1,0){24}}
  \put(-34,32){\line(1,0){24}}

  \put(-31,29){\line(1,-1){18}}
  \put(-31,11){\line(1,1){7.3}}
  \put(-13,29){\line(-1,-1){7.3}}

  \put(-6,19){$=e^{i\frac 3 4 \pi}($}

  \put(32,5){\line(0,1){30}}
  \put(22,5){\line(0,1){30}}

  \put(40,19){$-$}

  \put(62,23){\line(0,1){12}}
  \put(52,23){\line(0,1){12}}
  \put(52,23){\line(1,0){10}}
  \put(62,29){\circle*{2.}}
  \put(64,28){\tiny{$X$}}

  \put(52,5){\line(0,1){12}}
  \put(62,5){\line(0,1){12}}
  \put(52,17){\line(1,0){10}}
  \put(62,11){\circle*{2.}}
  \put(64,11){\tiny{$X$}}

  \put(70,19){$-$}

  \put(92,23){\line(0,1){12}}
  \put(82,23){\line(0,1){12}}
  \put(82,23){\line(1,0){10}}
  \put(92,29){\circle*{2.}}
  \put(94,28){\tiny{$R_{2\phi}$}}

  \put(82,5){\line(0,1){12}}
  \put(92,5){\line(0,1){12}}
  \put(82,17){\line(1,0){10}}
  \put(92,11){\circle*{2.}}
 \put(94,10){\tiny{$R^\dag_{2\phi}$}}

  \put(95,19){$-e^{-i\phi}$}

  \put(127,23){\line(0,1){12}}
  \put(117,23){\line(0,1){12}}
  \put(117,23){\line(1,0){10}}
  \put(127,29){\circle*{2.}}
  \put(129,28){\tiny{$R_{2\phi}$}}

  \put(117,5){\line(0,1){12}}
  \put(127,5){\line(0,1){12}}
  \put(117,17){\line(1,0){10}}
  \put(127,11){\circle*{2.}}
  \put(129,10){\tiny{$X$}}

  \put(133,19){$+e^{i\phi}$}

  \put(162,23){\line(0,1){12}}
  \put(152,23){\line(0,1){12}}
  \put(152,23){\line(1,0){10}}
  \put(162,29){\circle*{2.}}
  \put(164,28){\tiny{$X$}}

  \put(152,5){\line(0,1){12}}
  \put(162,5){\line(0,1){12}}
  \put(152,17){\line(1,0){10}}
  \put(162,11){\circle*{2.}}
  \put(164,10){\tiny{$R^\dag_{2\phi}$}}

  \put(169,19){$)$}

  \end{picture}
  \end{array}
  \en
  which together with (\ref{TL B other form 1}) point out the fact that no transparent topological deformations exist between such two configurations,
  although they are algebraically equivalent.

\section{How to solve the constraint relation (\ref{relation construction tangle relation constraint 1})}

\label{appendix trial solution tangle relation}

For example, we make a sketch on how to derive the representation of the BMW algebra, such as (\ref{representation tangle relation E 1}) and
(\ref{representation tangle relation U 1}) or (\ref{representation tangle relation U 2}) from the constraint relation
(\ref{relation construction tangle relation constraint 1}). Set the unitary bases $U_{ij}$ as $U_{ij}=X^i Z^j$ and the unitary base matrix $U_{mn}$
as $U_{mn}=1\!\!1_2$.
Then the constraint relation (\ref{relation construction tangle relation constraint 1}) has the form
\eq
\label{tangle relation Pauli matrix 1}
\frac 1 2 \sum_{k,l=0}^1\mu_{ij}\mu_{kl}X^kZ^lX^iZ^jZ^lX^k=X^iZ^j,
\en
which can be reformulated as
\eq
\label{tangle relation Pauli matrix 2}
\frac 1 2 \sum_{k,l=0}^1\mu_{ij}\mu_{kl}(-1)^{i\cdot l}(-1)^{k\cdot j}X^iZ^j=X^iZ^j.
\en
Since the unitary bases $U_{ij}$ satisfy the orthonormal relation (\ref{orthonormal condition U}), we have the constraint equation of the eigenvalues $\mu_{ij}$
of the Yang--Baxter gate $U$ (\ref{relation spectral theorem U gate}),
\eq
\label{tangle relation Pauli matrix 3}
\sum_{k,l=0}^1\mu_{ij}\mu_{kl}(-1)^{i\cdot l}(-1)^{k\cdot j}=2,
\en
which represent a set of equations given by
\eqa
  \mu_{00}\left(\mu_{00}+\mu_{01}+\mu_{10}+\mu_{11}\right)&=&2; \\
  \mu_{01}\left(\mu_{00}+\mu_{01}-\mu_{10}-\mu_{11}\right)&=&2; \\
  \mu_{10}\left(\mu_{00}-\mu_{01}+\mu_{10}-\mu_{11}\right)&=&2; \\
  \mu_{11}\left(\mu_{00}-\mu_{01}-\mu_{10}+\mu_{11}\right)&=&2.
  \ena
Solving the above equations, we have three classes of solutions for the eigenvalues $\mu_{ij}$ as below.
\begin{itemize}
  \item Class 1: $\mu_{00}=e^{i\phi}$, $\mu_{01}=e^{-i\phi}$, $\mu_{10}=e^{-i\phi}$, $\mu_{11}=-e^{i\phi}$, which determine
  the Yang--Baxter gate $U$ (\ref{relation spectral theorem U gate}) as
  \eq
     U=\sum_{i,j=0}^1\mu_{ij}|\psi(ij)\rangle\langle\psi(ij)|=\left(
                                                                \begin{array}{cccc}
                                                                  \cos\phi & 0 & 0 & i\sin\phi \\
                                                                  0 & -i\sin\phi & \cos\phi & 0 \\
                                                                  0 & \cos\phi & -i\sin\phi & 0 \\
                                                                  i\sin\phi & 0 & 0 & \cos\phi \\
                                                                \end{array}
                                                              \right).
    \en

  \item Class 2: $\mu_{00}=e^{i\phi}$, $\mu_{01}=e^{-i\phi}$, $\mu_{10}=-e^{i\phi}$, $\mu_{11}=e^{-i\phi}$, which determine
  the Yang--Baxter gate $U$ (\ref{relation spectral theorem U gate}) as
  \eq
     U=\sum_{i,j=0}^1\mu_{ij}|\psi(ij)\rangle\langle\psi(ij)|=\left(
                                                                \begin{array}{cccc}
                                                                  \cos\phi & 0 & 0 & i\sin\phi \\
                                                                  0 & -i\sin\phi & -\cos\phi & 0 \\
                                                                  0 & -\cos\phi & -i\sin\phi & 0 \\
                                                                  i\sin\phi & 0 & 0 & \cos\phi \\
                                                                \end{array}
                                                              \right).
  \en

  \item Class 3: $\mu_{00}=e^{i\phi}$, $\mu_{01}=-e^{i\phi}$, $\mu_{10}=e^{-i\phi}$, $\mu_{11}=e^{-i\phi}$, which determine
  the Yang--Baxter gate $U$ (\ref{relation spectral theorem U gate}) as
  \eq
     U=\sum_{i,j=0}^1\mu_{ij}|\psi(ij)\rangle\langle\psi(ij)|=\left(
                                                                \begin{array}{cccc}
                                                                  0 & 0 & 0 & e^{i\phi} \\
                                                                  0 & e^{-i\phi} & 0 & 0 \\
                                                                  0 & 0 & e^{-i\phi} & 0 \\
                                                                  e^{i\phi} & 0 & 0 & 0 \\
                                                                \end{array}
                                                              \right).
  \en
\end{itemize}

Similarly, when the unitary bases $U_{mn}$ are the Pauli gates $X$, $Z$ and $XZ$, respectively, the other solutions of the Yang--Baxter gate $U$ (\ref{relation spectral theorem U gate}) can be obtained. The collection of all the solutions has been presented in (\ref{representation tangle relation E 1}) and
(\ref{representation tangle relation U 1}) or (\ref{representation tangle relation U 2}), or (\ref{representation tangle relation E 2}) and
(\ref{representation tangle relation U 3}) or (\ref{representation tangle relation U 4}).

\section{Reformulation of the constraint relations (\ref{relation construction tangle relation constraint 1})-(\ref{relation construction tangle relation constraint 4}) and (\ref{relation construction general tangle relation constraint 1})-(\ref{relation construction general tangle relation constraint 4}) }

\label{skew transpose}

Both the constraint relations (\ref{relation construction tangle relation constraint 1})-(\ref{relation construction tangle relation constraint 4}) and the constraint relations (\ref{relation construction general tangle relation constraint 1})-(\ref{relation construction general tangle relation constraint 4}) looking complicated, we introduce the new conventions to simplify their formulations.
We define the skew-transpose on the product of two matrices as
\eq
(B\, C)^{ST}\equiv B^T\,C^T,
\en
where the skew-transposition $ST$  does not interchange $B^T$ and $C^T$ as the ordinary transpose does. With the new notations given by
\eq
\label{definition three operator in constraint relation}
\eta_{ijkl}\equiv\frac 1 2\mu_{ij}\mu_{kl}; \quad O_{\alpha\beta}\equiv U^\dag_{mn}U_{\alpha\beta},
\en
with specified indices $m$ and $n$, the constraint relations (\ref{relation construction tangle relation constraint 1})-(\ref{relation construction tangle relation constraint 4}) have the simplified forms
\eqa
\label{relation simplifed constraint 1}\sum_{k,l=0}^1\eta_{ijkl}O_{kl}{O^{ST}_{ij}}^\dag O^\dag_{kl} &=& {O^{ST}_{ij}}^\dag; \\
\label{relation simplifed constraint 2}\sum_{k,l=0}^1\eta_{ijkl}O^{ST}_{kl}O^\dag_{ij}{O_{kl}^{ST}}^\dag &=& O^\dag_{ij}; \\
\label{relation simplifed constraint 3}\sum_{k,l=0}^1\eta_{ijkl}O_{kl}O^{ST}_{ij} O_{kl}^\dag &=& O^{ST}_{ij}; \\
\label{relation simplifed constraint 4}\sum_{k,l=0}^1\eta_{ijkl}O^{ST}_{kl}O_{ij}{O_{kl}^{ST}}^\dag &=& O_{ij},
\ena
where the skew-transpose $ST$ is commutative with the Hermitian conjugation $\dag$.
As a remark, the notation $O_{\alpha\beta}$ is introduced to remove the indices $m$ and $n$
so that the algebraic structure of the constraint relations (\ref{relation construction tangle relation constraint 1})-(\ref{relation construction tangle relation constraint 4})
is presented in a more transparent way. Furthermore, with the new notation
\eq
\label{definition three operator in constraint relation}
\eta_{i_1j_1i_2j_2k_1l_1k_2l_2}\equiv\frac 1 2 G_{i_1j_1,k_1l_1}\tilde G_{i_2j_2,k_2l_2},
\en
the constraint relations (\ref{relation construction general tangle relation constraint 1})-(\ref{relation construction general tangle relation constraint 4}) have more simplified forms
\eqa
\label{relation general simplifed constraint 1}\sum_{k,l=0}^1\eta_{i_1j_1i_2j_2k_1l_1k_2l_2}O_{i_2j_2}{O^{ST}_{k_1l_1}}^\dag O^\dag_{k_2l_2} &=& {O^{ST}_{i_1j_1}}^\dag; \\
\label{relation general simplifed constraint 2}\sum_{k,l=0}^1\eta_{i_1j_1i_2j_2k_1l_1k_2l_2}O^{ST}_{i_2j_2}O^\dag_{k_1l_1}{O_{k_2l_2}^{ST}}^\dag &=& O^\dag_{i_1j_1}; \\
\label{relation general simplifed constraint 3}\sum_{k,l=0}^1\eta_{i_1j_1i_2j_2k_1l_1k_2l_2}O_{i_2j_2}O^{ST}_{k_1l_1}O_{k_2l_2}^\dag &=& O^{ST}_{i_1j_1}; \\
\label{relation general simplifed constraint 4}\sum_{k,l=0}^1\eta_{i_1j_1i_2j_2k_1l_1k_2l_2}O^{ST}_{i_2j_2}O_{k_1l_1}{O_{k_2l_2}^{ST}}^\dag &=& O_{i_1j_1}.
\ena
As a concluding remark, we hope that such the above reformulations of the constraint relations (\ref{relation construction tangle relation constraint 1})-(\ref{relation construction tangle relation constraint 4}) and (\ref{relation construction general tangle relation constraint 1})-(\ref{relation construction general tangle relation constraint 4}) are meaningful and useful elsewhere.


\begin{thebibliography}{99}

 \bibitem{NC2011}  M.A. Nielsen and I.L. Chuang, {\it Quantum Computation and Quantum Information}
 (Cambridge University Press, Cambridge, UK, 2000 and 2011).

  \bibitem{Preskill97}  J. Preskill, {\it Lecture Notes on Quantum Computation},
  http://www.theory.caltech.edu/preskill.

  \bibitem{BBCCJPW93} C.H. Bennett, G. Brassard, C. Crepeau, R.
  Jozsa, A. Peres and W.K. Wootters, {\it  Teleporting an Unknown Quantum State
  via Dual Classical and Einstein-Podolsky-Rosen Channels}, Phys. Rev. Lett. {\bf 70}, 1895 (1993).

  \bibitem{Vaidman94} L. Vaidman, {\it Teleportation of Quantum States},
 Phys. Rev.  A {\bf 49}, 1473-1475 (1994).

 \bibitem{BDM00} S.L. Braunstein, G.M. D'Ariano, G.J. Milburn
 and M.F. Sacchi, {\it Universal Teleportation with a Twist}, Phys.
 Rev. Lett. {\bf 84}, 3486-3489 (2000).

 \bibitem{Werner01} R.F. Werner,  {\it All Teleportation and Dense Coding
 Schemes},  J. Phys. A: Math. Theor.  {\bf 35}, 7081-7094 (2001).

 \bibitem{Aravind97} P.K. Aravind, {\it Borromean Entanglement of the GHZ State}, in {\it Potentiality,
 Entanglement and Passion-at-a-Distance}, Springer Netherlands, 53-59 (1997).

  \bibitem{KL02} L.H. Kauffman and S.J. Lomonaco Jr., {\it Quantum Entanglement and Topological Entanglement},
   New J. Phys. {\bf 4}, 73 (2002).

  \bibitem{Kauffman02} L.H. Kauffman, {\it Knots and Physics}
  (World Scientific Publishers, 2002).

   \bibitem{YBE67} C.N. Yang, {\it Some Exact Results for the Many Body Problems in One Dimension with Repulsive Delta
 Function Interaction}, Phys. Rev. Lett. {\bf 19}, 1312-1314 (1967). R.J. Baxter, {\it Partition Function of
 the Eight-Vertex Lattice Model}, Ann. Phys. {\bf 70}, 193-228 (1972).  J.H.H. Perk and H. Au-Yang, {\it Yang--Baxter
 Equations}, Encyclopedia of Mathematical Physics, Vol. 5, 465-473 (Elsevier Science, Oxford, 2006).


  \bibitem{Dye03} H. Dye, {\it Unitary Solutions to the Yang--Baxter Equation in Dimension Four},
 Quantum Inf. Process. {\bf 2}, 117-150 (2003).

  \bibitem{KL04} L.H. Kauffman and S.J. Lomonaco Jr.,  {\it Braiding Operators are Universal Quantum Gates},
  New J. Phys. {\bf 6}, 134 (2004).

 \bibitem{ZKG04}  Y. Zhang, L.H. Kauffman and M.L. Ge, {\it Universal Quantum Gate,
 Yang--Baxterization and Hamiltonian}, Int. J. Quant. Inform. {\bf 4}, 669-678 (2005).

 \bibitem{AJJ16} G. Alagic, M. Jarret and S.P. Jordan, {\it Yang--Baxter Operators Need Quantum Entanglement
 to Distinguish Knots}, J. Phys. A: Math. Theor. {\bf 49}, 075203 (2016).

  \bibitem{TL71} H.N.V. Temperley and E.H. Lieb, {\it Relations between
  the `Percolation' and `Colouring' Problem and Other Graph-Theoretical Problems
  Associated with Regular Planar Lattices: Some Exact Results for the `Percolation'
  Problem}, Proc. Roy. Soc. A {\bf 322}, 251 (1971).

  \bibitem{BW89} J.S. Birman and H.B. Wenzl, {\it Link Polynomials and a New Algebra}, Transactions of the American
  Mathematical Society {\bf 313}, 249-273 (1989). J. Murakami, {\it The Kauffman Polynomial of Links and Representation
  Theory}, Osaka J. Math {\bf 24}, 745-758 (1987).

  \bibitem{Jones89} V. Jones, {\it On a Certain Value of the Kauffman Polynomial}, Commun. Math. Phys. {\bf125}, 459-467 (1989).

  \bibitem{CGX91} Y. Cheng, M.L. Ge and K. Xue, {\it Yang--Baxterization of Braid Group Representations},
  Commun. Math. Phys. {\bf136}, 195-208 (1991).

  \bibitem{WXSZHW10} G. Wang, K. Xue, C. Sun, C. Zhou, T. Hu and Q. Wang, {\it Temperley--Lieb Algebra,
  Yang--Baxterization and Universal Gate}, Quantum Inf. Process. {\bf 9},  699-710 (2010).

 \bibitem{Zhang06}  Y. Zhang, {\it Teleportation, Braid Group and
   Temperley--Lieb Algebra},  J. Phys. A:  Math. Theor.  {\bf 39}, 11599-11622 (2006).
    Y. Zhang and L.H. Kauffman,
  {\it Topological-Like Features in Diagrammatical Quantum
  Circuits}, Quantum Inf. Process. {\bf 6}, 477-507 (2007).
  Y. Zhang, {\it Braid Group, Temperley--Lieb Algebra,
   and Quantum Information and Computation}, AMS Contemporary Mathematics {\bf 482}, 52 (2009).

 \bibitem{ZZP15} Y. Zhang, K. Zhang and J.-L. Pang, {\it Teleportation-Based Quantum
  Computation, Extended Temperley--Lieb Diagrammatical Approach and Yang--Baxter Equation},
   Quantum Inf. Process. {\bf 15}, 405-464 (2016).

  \bibitem{WXSLZ15} G. Wang, K. Xue, C. Sun, B. Liu, Y. Liu, and Y. Zhang, {\it  Topological Basis
  Associated with B-M-W algebra: Two Spin-1/2 Realization},
   Phys. Lett. A {\bf 379}, 1-4 (2015).

   \bibitem{BMW11} C. Zhou, K. Xue, G. Wang, C. Sun and G. Du, {\it Birman-Wenzl-Murakami Algebra and Topological Basis},
  Communications in Theoretical Physics {\bf 57}, 179 (2012). C. Zhou, K. Xue, L. Gou, C. Sun, G. Wang and T. Hu,
  {\it Birman-Wenzl-Murakami Algebra, Topological Parameter and Berry Phase}, Quantum Inf. Process. {\bf 11}, 1765-1773 (2012).
   Q. Zhao, R.Y. Zhang, K. Xue and M.L. Ge, {\it Topological Basis Associated with BWMA, Extremes of $L_1$-norm in
   Quantum Information and Applications in Physics}, arXiv:1211.6178 (2012).


  \bibitem{GC99} D. Gottesman and I.L. Chuang, {\it Demonstrating the Viability of Universal Quantum
   Computation Using Teleportation and Single-Qubit Operations}, Nature {\bf 402}, 390 (1999).

 \bibitem{Nielsen03} M.A. Nielsen,
 {\it Universal Quantum Computation Using Only Projective
  Measurement, Quantum Memory, and Preparation of the $0$ State},
  Phys. Lett. A \textbf{308}, 96 (2003).

 \bibitem{Brauer37} R. Brauer, {\it On Algebras Which are Connected With the Semisimple Continuous Groups},
 Ann. of Math. {\bf 38}, 857-872 (1937).

  \bibitem{ZZ14} Y. Zhang and K. Zhang, {\it GHZ transform (I):  Bell transform and quantum teleportation}, arXiv:1401.7009 (2014).

 \bibitem{VW00} K.G.H. Vollbrecht and R.F. Werner, {\it Why Two Qubits are Special}, Journal of Mathematical
 Physics {\bf 41}, 6772-6782 (2000).

 \bibitem{KC02} B. Kraus and J.I. Cirac, {\it Optimal Creation of Entanglement Using a Two-Qubit Gate},
 Phys. Rev. A {\bf 63}, 062309 (2001).

 \bibitem{ZZF00} P. Zanardi, C. Zalka and L. Faoro,  {\it Entangling Power of Quantum Evolutions},  Phys. Rev. A {\bf 62}, 030301 (2000).

 \bibitem{SBM04} V.V. Shende, S.S. Bullock and I.L. Markov, {\it Recognizing Small-Circuit Structure in
 Two-Qubit Operators}, Phys. Rev. A {\bf 70}, 012310 (2004).

 \bibitem{Gottesman97} D. Gottesman, {\it Stabilizer Codes and Quantum Error Correction Codes},
 Ph.D. Thesis, CalTech, Pasadena, CA, 1997.

 \bibitem{BMPRV00} P.O. Boykin, T. Mor, M. Pulver, V. Roychowdhury and F. Vatan,  {\it A New Universal and
 Fault-Tolerant Quantum Basis},
 Inf. Process. Lett, {\bf 75}, 101-107 (2000).

 \bibitem{PBR15} A. Pourkia, J. Batle and C.H. Raymond Ooi, {\it Cyclic Groups and Quantum Logic Gates}, arXiv:1509.08252 (2015).

 \end{thebibliography}
\end{document}